\documentclass{lmcs}
\pdfoutput=1

\usepackage{lastpage}
\lmcsdoi{15}{1}{31}
\lmcsheading{}{\pageref{LastPage}}{}{}%
{Nov.~02,~2017}{Mar.~29,~2019}{}

\usepackage[utf8]{inputenc}

\usepackage{amsmath,mathtools,amssymb,stmaryrd,eufrak}
\usepackage{xspace}
\usepackage{tabularx}
\usepackage{mathpartir,url}

\usepackage{float}
\floatstyle{boxed}
\restylefloat{figure}

\theoremstyle{defC}
\newtheorem{exaC}[thm]{Example}

\usepackage{nf}

\begin{document}

\title[Proving Soundness of Extensional Normal-Form
  Bisimilarities]{Proving Soundness\texorpdfstring{\\}{} of Extensional Normal-Form
  Bisimilarities}

\author[]{Dariusz Biernacki\rsuper{a}}

\author[]{Sergue\"{\i} Lenglet\rsuper{b}}

\author[]{Piotr Polesiuk\rsuper{a}}

\address{\lsuper{a}University of Wroc\l{}aw, Wroc\l{}aw, Poland}
\email{dabi@cs.uni.wroc.pl}
\email{ppolesiuk@cs.uni.wroc.pl}

\address{\lsuper{b}Universit\'e de Lorraine, Nancy, France}
\email{serguei.lenglet@univ-lorraine.fr}

\thanks{This is a revised and extended version
  of~\cite{Biernacki-al:MFPS17}. This work was supported by PHC Polonium and by
  National Science Centre, Poland, grant no. 2014/15/B/ST6/00619.}



\begin{abstract}
  Normal-form bisimilarity is a simple, easy-to-use behavioral equivalence that
  relates terms in $\lambda$-calculi by decomposing their normal forms into
  bisimilar subterms. Moreover, it typically allows for powerful up-to
  techniques, such as bisimulation up to context, which simplify bisimulation
  proofs even further. However, proving soundness of these relations becomes
  complicated in the presence of $\eta$-expansion and usually relies on ad hoc
  proof methods which depend on the language. In this paper we propose a more
  systematic proof method to show that an extensional normal-form bisimilarity
  along with its corresponding up to context technique are sound. We illustrate
  our technique with three calculi: the call-by-value $\lambda$-calculus, the
  call-by-value $\lambda$-calculus with the delimited-control operators
  \textshift{} and \textreset{}, and the call-by-value $\lambda$-calculus with
  the abortive control operators \textcallcc and \textabort. In the first two
  cases, there was previously no sound up to context technique validating the
  $\eta$-law, whereas no theory of normal-form bisimulations for a calculus with
  \textcallcc and \textabort has been presented before. Our results have been
  fully formalized in the Coq proof assistant.
\end{abstract}

\maketitle

\section{Introduction}%
\label{s:intro}

In formal languages inspired by the $\lambda$-calculus, the behavioral
equivalence of choice is usually formulated as a Morris-style contextual
equivalence~\cite{JHMorris:PhD}: two terms are equivalent if they behave the
same in any context. This criterion captures quite naturally the idea that
replacing a term by an equivalent one in a bigger program should not affect the
behavior of the whole program. However, the quantification over contexts makes
contextual equivalence hard to use in practice to prove the equivalence of two
given terms. Therefore, it is common to look for easier-to-use, \emph{sound}
alternatives that are at least included in contextual equivalence, such as
coinductively defined \emph{bisimilarities}.

Different styles of bisimilarities have been defined for the $\lambda$-calculus,
including \emph{applicative bisimilarity}~\cite{Abramsky-Ong:IaC93},
\emph{normal-form bisimilarity}~\cite{Lassen:LICS05} (originally called
\emph{open bisimilarity} in~\cite{Sangiorgi:LICS92}), and \emph{environmental
  bisimilarity}~\cite{Sangiorgi-al:TOPLAS11}. Applicative and environmental
bisimilarities compare terms by applying them to function arguments, while
normal-form bisimilarity reduces terms to normal forms, which are then
decomposed into bisimilar subterms. As we can see, applicative and environmental
bisimilarities still rely on some form of quantification over arguments, which
is not the case of normal-form bisimilarity. As a drawback, the latter is
usually not \emph{complete} w.r.t.\ contextual equivalence---there exist
contextually equivalent terms that are not normal-form bisimilar---while the
former are. Like environmental bisimilarity, normal-form bisimilarity usually
allows for \emph{up-to techniques}~\cite{Sangiorgi-Pous:11}, relations which
simplify equivalence proofs of terms by having less requirements than regular
bisimilarities. For example, reasoning up to context allows to forget about a
common context: to equate $\inctx C t$ and $\inctx C s$, it is enough to relate
$t$ and $s$ with a bisimulation up to context.

In the call-by-value $\lambda$-calculus, the simplest definition of normal-form
bisimilarity compares values by equating a variable only with itself, and a
$\lambda$-abstraction only with a $\lambda$-abstraction such that their bodies
are bisimilar. Such a definition does not respect call-by-value
$\eta$-expansion, since it distinguishes $x$ from $\lam y {\app x y}$. A less
discriminating definition instead compares values by applying them to a fresh
variable, thus relating $\lam y {\app v y}$ and $v$ for any value~$v$ such that
$y$ is not free in $v$: given a fresh $z$, $(\lam y {\app v y}) \iapp z$ reduces
to $v \iapp z$. Such a bisimilarity, that we call \textit{extensional
  bisimilarity},\footnote{Lassen uses the term \textit{bisimilarity up to
    $\eta$}~\cite{Lassen:MFPS99} for a normal-form bisimilarity that validates
  the $\eta$-law, but we prefer the term \textit{extensional bisimilarity} so that
  there is no confusion with notions referring to up-to techniques such as
  bisimulation up to context.}  relates more contextually equivalent terms, but
proving its soundness as well as proving the soundness of its up-to techniques
is more difficult, and usually requires ad hoc proof methods, as we detail in
the related work section (Section~\ref{s:relwork}).

Madiot et al.~\cite{Madiot-al:CONCUR14} propose a framework where proving the
soundness of up-to techniques is quite uniform and simpler. It also allows to
factorize proofs, since showing that reasoning up to context is sound
directly implies that the corresponding bisimilarity is a congruence, which is
the main property needed for proving its soundness. Madiot et al.\ apply the
method to environmental bisimilarities for the plain call-by-name
$\lambda$-calculus and for a call-by-value $\lambda$-calculus with references,
as well as to a bisimilarity for the $\pi$-calculus. In a subsequent
work~\cite{Aristizabal-al:FSCD16}, we extend this framework to define
environmental bisimilarities for a call-by-value $\lambda$-calculus with
multi-prompted delimited-control operators. We propose a distinction between
strong and regular up-to techniques, where regular up-to techniques cannot be
used in certain bisimilarity tests, while strong ones can always be used. This
distinction allows to prove sound more powerful up-to techniques, by forbidding
their use in cases where it would be unsound to apply them.

So far, the method developed in~\cite{Madiot-al:CONCUR14,Aristizabal-al:FSCD16}
have been used in the $\lambda$-calculus only for environmental
bisimilarities. In this paper, we show that our extended
framework~\cite{Aristizabal-al:FSCD16} can also be used to prove the soundness
of extensional normal-form bisimilarities and their corresponding bisimulation
up to context. We first apply it to the plain call-by-value $\lambda$-calculus,
in which an extensional normal-form bisimilarity, albeit without a corresponding
bisimulation up to context, have already been proved sound~\cite{Lassen:LICS05},
to show how our framework allows to prove soundness for both proof techniques at
once. We then consider a call-by-value $\lambda$-calculus with the
delimited-control operators \textshift{} and
\textreset~\cite{Danvy-Filinski:LFP90}, for which there has been no sound
bisimulation up to context validating the $\eta$-law either, and we show that
our method applies seamlessly in that setting as well. Finally, we address a
calculus of abortive control, i.e., the call-by-value $\lambda$-calculus with
\textcallcc and \textabort~\cite{Felleisen-Friedman:FDPC3,Felleisen-Hieb:TCS92}.
If a normal-form bisimilarity has been defined for the
$\lambda\mu$-calculus~\cite{Stoevring-Lassen:POPL07}, a more expressive calculus
with abortive control, no theory of normal-form bisimulations has been proposed
so far for \textcallcc and \textabort themselves. This confirms the robustness
of the proof method we advocate in this article, but it also provides a new
operational technique for reasoning about classical calculi of abortive
continuations introduced by Felleisen et al.

Our results have been fully formalized in the Coq proof assistant.
The goal of the formalization is twofold: to increase our confidence in the
correctness of the proof method, but also to leave most of the routine and
repeatable compatibility proofs outside of the article proper. While the proof
method is introduced and illustrated in detail in the main body of the paper,
the interested reader can consult the formalization, commented and with pointers
to the paper, for the complete development (excluding the practical examples
that are presented here for illustration purposes only). The Coq formalization,
available at~\url{https://bitbucket.org/pl-uwr/diacritical}, uses a de Bruijn
representation for $\lambda$-terms, where the de Bruijn indices are encoded
using nested datatypes~\cite{Bird-Paterson:JFP99}.

The paper is organized as follows: in Section~\ref{s:relwork}, we discuss the
previous proofs of soundness of extensional normal-form bisimilarities. In
Section~\ref{s:lamcal}, we present the proof method for the call-by-value
$\lambda$-calculus, that we then apply to the $\lambda$-calculus with delimited
control in Section~\ref{s:delcon}, and to the $\lambda$-calculus with abortive
control in Section~\ref{s:abortcon}. We conclude in Section~\ref{s:conclusion}.
Compared to the conference article~\cite{Biernacki-al:MFPS17},
Section~\ref{s:abortcon} is entirely new, whereas minor revisions have been done
to the remaining sections.

\section{Related Work}%
\label{s:relwork}

Normal-form bisimilarity has been first introduced by
Sangiorgi~\cite{Sangiorgi:LICS92} and has then been defined for many variants of
the $\lambda$-calculus, considering
$\eta$-expansion~\cite{Lassen:MFPS99,Lassen:LICS05,Lassen:LICS06,
  Stoevring-Lassen:POPL07,Lassen-Levy:CSL07,Lassen-Levy:LICS08,
  Biernacki-Lenglet:FLOPS12,Biernacki-al:HAL15} or
not~\cite{Lassen:99,Lassen:MFPS05}. In this section we focus on the articles
treating the $\eta$-law, and in particular on the congruence and soundness
proofs presented therein.

In~\cite{Lassen:MFPS99}, Lassen defines several equivalences for the
call-by-name $\lambda$-calculus, depending on the chosen semantics. He defines
\emph{head-normal-form (hnf) bisimulation} and \emph{hnf bisimulation up
  to~$\eta$} for the semantics based on reduction to head normal form (where
$\eta$-expansion applies to any term $t$, not only to a value as in the
call-by-value $\lambda$-calculus), and \emph{weak-head-normal-form (whnf)
  bisimulation} based on reduction to weak head normal form. (It does not make
sense to consider a \emph{whnf bisimulation up to $\eta$}, since it would be
unsound, e.g., it would relate a non-terminating term $\Omega$ with a normal
form $\lam{x}{\app{\Omega}{x}}$.)
The paper also defines a bisimulation up to context for each
bisimilarity.

The congruence proofs for the three bisimilarities follow from the main lemma
stating that if a relation is a bisimulation, then so is its substitutive and
context closure. The lemma is proved by nested induction on the definition of
the closure and on the number of steps in the evaluation of terms to normal
forms. It can be easily strengthened to prove the soundness of a bisimulation up
to context: if a relation is a bisimulation up to context, then its substitutive
and context closure is a bisimulation.  The nested induction proof method has
been then applied to prove congruence for a whnf bisimilarity for the
call-by-name $\lambda\mu$-calculus~\cite{Lassen:99} (a calculus with
continuations), an extensional hnf bisimilarity for the call-by-name
$\lambda$-calculus with pairs~\cite{Lassen:LICS06}, and a whnf bisimilarity for
a call-by-name $\lambda$-calculus with McCarthy's ambiguous choice ({\tt amb})
operator~\cite{Lassen:MFPS05}. These papers do not define any corresponding
bisimulation up to context.

Lassen uses another proof technique in~\cite{Lassen:LICS05}, where he defines an
\emph{eager normal form (enf) bisimilarity} and an \emph{enf bisimilarity up to
  $\eta$}.\footnote{While weak head normal forms are normal forms under
  call-by-name evaluation, eager normal forms are normal forms under
  call-by-value evaluation of $\lambda$-terms.} Lassen shows that the
bisimilarities correspond to B{\"o}hm trees equivalence (up to $\eta$) after a
continuation-passing style (CPS) translation, and then he deduces congruence of
the enf bisimilarities from the congruence of the B{\"o}hm trees equivalence. A
CPS-translation based technique has also been used in~\cite{Lassen:LICS06} to
prove congruence of the extensional bisimilarity for the call-by-name
$\lambda$-calculus (also with surjective pairing), the $\lambda\mu$-calculus,
and the $\Lambda\mu$-calculus. Unlike the nested induction proof method, this
technique does not extend to a soundness proof of a bisimulation up to context.

In~\cite{Lassen:LICS05}, Lassen claims that \textit{``It is also possible to prove
  congruence of enf bisimilarity and enf bisimilarity up to $\eta$ directly like
  the congruence proofs for other normal form bisimilarities (tree equivalences)
  in~\cite{Lassen:MFPS99}, although the congruence proofs (\dots)  require
  non-trivial changes to the relational substitutive context closure operation
  in op.cit.\ (\dots)  Moreover, from the direct congruence proofs, we can derive
  bisimulation “up to context” proof principles like those for other normal form
  bisimilarities in op.cit.''} To our knowledge, such a proof is not published
anywhere; we tried to carry out the congruence proof by following this comment,
but we do not know how to conclude in the case of enf bisimilarity up to
$\eta$. We discuss what the problem is at the end of the proof of
Lemma~\ref{l:app-lambda}.

St{\o}vring and Lassen~\cite{Stoevring-Lassen:POPL07} define extensional enf
bisimilarities for three calculi: $\lambda\mu$ (continuations), $\lambda\rho$
(mutable state), and $\lambda\mu\rho$ (continuations and mutable state). The
congruence proof is rather convoluted and is done in two stages: first, prove
congruence of a non-extensional bisimilarity using the nested induction
of~\cite{Lassen:MFPS99}, then extend the result to the extensional bisimilarity
by a syntactic translation that takes advantage of an infinite $\eta$-expansion
combinator. The paper does not mention bisimulation up to context.

Lassen and Levy~\cite{Lassen-Levy:CSL07,Lassen-Levy:LICS08} define a normal-form
bisimilarity for a CPS calculus called JWA equipped with a rich type system
(including product, sum, recursive types;~\cite{Lassen-Levy:LICS08} adds
existential types). The bisimilarity respects the $\eta$-law, and the congruence
proof is done in terms of game semantics notions. Again, these papers do not
mention bisimulation up to context.

In a previous work~\cite{Biernacki-Lenglet:FLOPS12}, we define extensional enf
bisimilarities and bisimulations up to context for a call-by-value
$\lambda$-calculus with delimited-control operators. The (unpublished)
congruence and soundness proofs follow Lassen~\cite{Lassen:MFPS99}, but are
incorrect: one case in the induction, that turns out to be problematic, has been
forgotten. In~\cite{Biernacki-al:HAL15} we fix the congruence proof of the
extensional bisimilarity, by doing a nested induction on a different notion of
closure than Lassen. This approach fails when proving soundness of a
bisimulation up to context, and therefore bisimulation up to context does not
respect the $\eta$-law in~\cite{Biernacki-al:HAL15}.

To summarize:
\begin{itemize}
\item[-] The soundness proofs for extensional hnf bisimilarities are uniformly
  done using a nested induction proof
  method~\cite{Lassen:MFPS99,Lassen:LICS06}. The proof can then be turned into a
  soundness proof for bisimulation up to context.
\item[-] The soundness proofs of extensional enf bisimilarities either follow
  from a CPS translation~\cite{Lassen:LICS05,Lassen:LICS06}, or other ad hoc
  arguments~\cite{Stoevring-Lassen:POPL07,Lassen-Levy:CSL07,Lassen-Levy:LICS08,
    Biernacki-al:HAL15} which do not carry over to a soundness proof for a
  bisimulation up to context.
\item[-] The only claims about congruence of an extensional enf bisimilarity as
  well as soundness of the corresponding bisimulation up to context using a
  nested induction proof are either wrong~\cite{Biernacki-Lenglet:FLOPS12} or
  are not substantiated by a presentation of the actual
  proof~\cite{Lassen:LICS05}. The reason the nested induction proof works for
  extensional hnf bisimilarities and not for extensional enf bisimilarities
  stems from the difference in the requirements on the shape of
  $\lambda$-abstractions the two normal forms impose: whereas the body of a
  $\lambda$-abstraction in hnf is also a hnf, the body of a
  $\lambda$-abstraction in enf is an arbitrary term.
\end{itemize}
In this paper, we consider an extensional enf bisimilarity for three calculi:
the plain $\lambda$-calculus and its extensions with delimited and abortive
continuations, and in each case we present a soundness proof of the
corresponding enf bisimulation up to context from which congruence of the
bisimilarity follows.

\section{Call-by-value \texorpdfstring{$\lambda$}{lambda}-calculus}%
\label{s:lamcal}

We introduce a new approach to normal-form bisimulations that is based on the
framework we developed previously~\cite{Aristizabal-al:FSCD16}. The calculus of
discourse is the plain call-by-value $\lambda$-calculus.

\subsection{Syntax, semantics, and normal-form bisimulations}%
\label{ss:lamcal-calculus}

We let $x$, $y$, $z$ range over variables. The syntax of terms ($t$, $s$),
values ($v$, $w$), and call-by-value evaluation contexts~($E$) is given as
follows:

\[
\begin{array}{rcl}
  t,s & \bnfdef & v \bnfor \app t s
  \\
  v, w & \bnfdef & x \bnfor \lam x t
  \\[1mm]
  E & \bnfdef & \hole \bnfor \app E t \bnfor \app v E
\end{array}
\]

\vspace{2mm}\noindent An abstraction $\lam x t$ binds $x$ in $t$; a variable
that is not bound is called free. The set of free variables in a term $t$ is
written $\fv t$. We work modulo $\alpha$-conversion of bound variables, and a
variable is called fresh if it does not occur in the terms under
consideration. Contexts are represented outside-in, and we write $\inctx E t$
for plugging a term in a context. We write $\subst t x v$ for the
capture-avoiding substitution of $v$ for~$x$ in $t$. We write successive
$\lambda$-abstractions $\lam x {\lam y t}$ as $\lam{xy} t$.

We consider a call-by-value reduction semantics for the language

\[
\inctx{E}{\app{(\lam x t)} v}  \rawred \inctx{E}{\subst t x v}
\]

\vspace{2mm}\noindent We write $\rtc\rawred$ for the reflexive and transitive
closure of $\rawred$, and $\redto t s$ if $\redrtc t s$ and $s$ cannot reduce;
we say that $t$ evaluates to $s$.

Eager normal forms are either values or \emph{open stuck terms} of the form
$\inctx E {\app x v}$.  Normal-form bisimilarity relates terms by comparing
their normal forms (if they exist).  For values, a first possibility is to
relate separately variables and $\lambda$-abstractions: a variable $x$ can be
equated only to $x$, and $\lam x t$ is bisimilar to $\lam x s$ if $t$ is
bisimilar to $s$. As explained in the introduction, this does not respect
$\eta$-expansion: the $\eta$-respecting definition compares values by applying
them to a fresh variable. Given a relation $\rel$ on terms, we reflect how
values and open stuck terms are tested by the relations~$\testval \rel$,
$\testevctx \rel$, and $\testopen \rel$, defined as follows:

\begin{mathpar}
  \inferrule{\app v x \rel \app w x \\ x \mbox{ fresh}}
  {v \testval\rel w}
  \and
  \inferrule{\inctx E x \rel \inctx{E'} x \\ x \mbox{ fresh}}
  {E \testevctx \rel E'}
  \and
  \inferrule{E \testevctx\rel E' \\ v \testval\rel w}
  {\inctx E {\app x v} \testopen\rel \inctx {E'}{\app x w}}
\end{mathpar}

\vspace{2mm}
\begin{rem}%
\label{r:abstrvars}
Traditionally, normal-form bisimulations are construed as an open version of
applicative bisimulations in that they test values by applying them to a free
variable~\cite{Lassen:LICS05}, rather than to all possible closed
values~\cite{Abramsky-Ong:IaC93}. However, a connection with B{\"o}hm or
L{\'e}vy-Longo trees~\cite{Lassen:MFPS99} aside, one could introduce a separate
category of variables that would represent \emph{abstract values}, and use these
for the purpose of testing functional values. In such an approach, the reduction
relation would cater for closed terms only, as far as the term variables are
concerned, and the notion of an open stuck term could be replaced with a notion
of a \emph{value-stuck term}. In this work we stick to the traditional approach
to testing functional values (witness the definition of $\testval{R}$), but in
Section~\ref{s:abortcon} we propose an extension which is analogous to the one
sketched in this remark, and we introduce a separate category of variables
representing \emph{abstract contexts}, a notion dual to that of abstract values.
\end{rem}

\vspace{2mm} We can now define (extensional) normal-form bisimulation and
bisimilarity, using a notion of progress.
\begin{defi}%
  \label{def:nf-classic}
  A relation $\rel$ progresses to $\rels$ if $t \rel s$ implies:
  \begin{itemize}
  \item if $\red t {t'}$, then there exists $s'$ such that $\redrtc s {s'}$ and
    $t' \rels s'$;
  \item if $t=v$, then there exists $w$ such that $\redto s w$, and $v
    \testval\rels w$;
  \item if $t=\inctx E {\app x v}$, then there exist $E'$, $w$ such that $\redto
    s {\inctx {E'}{\app x w}}$ and $\inctx E {\app x v} \testopen\rels \inctx
    {E'}{\app x w}$;
  \item the converse of the above conditions on $s$.
  \end{itemize}
\end{defi}

\noindent
A bisimulation is then defined as a relation which progresses to itself, and
bisimilarity as the union of all bisimulations. Our definition is in a
small-step style, unlike Lassen's~\cite{Lassen:LICS05}, as we believe small-step
is more flexible, since we can recover a big-step reasoning with up to reduction
(Section~\ref{ss:lamcal-up-to-nf}). In usual
definitions~\cite{Lassen:LICS05,Stoevring-Lassen:POPL07,Biernacki-al:HAL15}, the
$\beta$-reduction is directly performed when a $\lambda$-abstraction is applied
to a fresh variable, whereas we construct an application in order to uniformly
treat all kinds of values, and hence account for $\eta$-expansion. However, with
this approach a naive reasoning up to context would be unsound because it would
equate any two values: if $v$ and $w$ are related, then $\app v x$ and
$\app w x$ are related up to context.  In our framework, we prevent this issue
by not allowing this up-to technique in that case, as we now explain.

We recast the definition of normal-form bisimilarity in the framework of our
previous work~\cite{Aristizabal-al:FSCD16}, which is itself an extension of a
work by Madiot et al.~\cite{Madiot-al:CONCUR14,Madiot:PhD}. The goal is to
factorize the congruence proof of the bisimilarity with the soundness proofs of
the up-to techniques. The novelty in~\cite{Aristizabal-al:FSCD16} is that we
distinguish between \emph{active} and \emph{passive} clauses, and we forbid some
up-to techniques to be applied in a passive clause. Whereas this distinction
does not change the notions of bisimulation or bisimilarity, it has an impact on
the bisimilarity congruence proof.

\begin{defi}
  \label{def:progress}
  A relation $\rel$ \textit{diacritically progresses} to $\rels$, $\relt$ written
  $\rel \pprogress \rels, \relt$, if $\rel \mathop\subseteq \rels$, $\rels
  \mathop\subseteq \relt$, and $t \rel s$ implies:
  \begin{itemize}
  \item if $\red t {t'}$, then there exists $s'$ such that $\redrtc s {s'}$ and
    $t' \relt s'$;
  \item if $t=v$, then there exists $w$ such that $\redto s w$, and $v
    \testval\rels w$;
  \item if $t=\inctx E {\app x v}$, then there exist $E'$, $w$ such that $\redto
    s {\inctx {E'}{\app x w}}$ and $\inctx E {\app x v} \testopen\relt \inctx
    {E'}{\app x w}$;
    \item the converse of the above conditions on $s$.
  \end{itemize}
  A normal-form bisimulation is a relation $\rel$ such that $\rel \pprogress
  \rel, \rel$, and normal-form bisimilarity $\nfbisim$ is the union of all
  normal-form bisimulations.
\end{defi}
The difference between Definitions~\ref{def:progress} and~\ref{def:nf-classic}
is only in the clause for values, where we progress towards a different relation
than in the other clauses of Definition~\ref{def:progress}. We say that the
clause for values is passive, while the others are active. A bisimulation $\rel$
progresses towards $\rel$ in passive and active clauses, so the two notions of
bisimulation coincide: $\rel$ is a bisimulation according to
Definition~\ref{def:progress} iff it is a bisimulation by
Definition~\ref{def:nf-classic}. Consequently, the resulting bisimilarity is
also the same between the two definitions.

\begin{exa}%
  \label{ex:wads}
  Let $\theta \defeq \lam{zx}{x \iapp \lam y {z \iapp z \iapp x \iapp y}}$ and
  $\fix v \defeq \lam x {\theta \iapp \theta \iapp v \iapp x}$ for a given $v$;
  note that $\redrtc{\fix v \iapp x}{v \iapp \fix v \iapp x}$. Wadsworth's
  infinite $\eta$-expansion combinator~\cite{Barendregt:84} can be defined as
  $\wads \defeq \fix{\lam {fxy}{x \iapp {(f \iapp y)}}}$. Let
  $\idbis \defeq \{ (t, t) \mid t \mbox{ any term} \}$ be the
  identity bisimulation. We prove that $\lam x x \nfbisim \wads$, by showing
  that
  \begin{align*}
    \mathord{\rel} \defeq \idbis \cup \{(\lam x x, \wads) \}
    &  \cup \{ (t, s) \mid \redrtc{(\lam x x)
                                   \iapp y} t, \redrtc{\app \wads y} s \bnfor y \mbox{
                                   fresh} \} \\
                                 & \cup \{ (y \iapp z, t) \mid \redrtc{(\lam x {y
                                   \iapp {(\app \wads x)}})
                                   \iapp z} t \bnfor y, z \mbox{
                                   fresh} \}
  \end{align*}
  is a bisimulation. Indeed, to compare $\lam x x$ and $\wads$, we have to
  relate $\app {(\lam x x)} y$ and $\app \wads y$, but $\redrtc{\app \wads
    y}{\lam x {y \iapp {(\app \wads x)}}}$. We then have to equate $\app y z$
  and $(\lam x {y \iapp {(\app \wads x)}}) \iapp z$, the latter evaluating to $y
  \iapp {\lam x {z \iapp {(\app \wads x)}}}$. To relate these open stuck terms,
  we have to equate~$\hole$ and $\hole$ (with $\idbis$), and~$z$ with $\lam x {z
    \iapp {(\app \wads x)}}$, but these terms are already in $\rel$.
\end{exa}

The purpose of the distinction between active and passive is not to change
regular bisimulation proofs, but instead to affect how up-to techniques are
applied, by forbidding the use of some of them in a passive clause. In
particular, reasoning up to context is not allowed in the value case, meaning
that we cannot use the fact that $\app v x$ and $\app w x$ are related up to
context whenever $v$ and $w$ are related to conclude a proof with bisimulation
up to context. In contrast, it is safe to use any (combination of) up-to
techniques in an active clause. In the next section, we explain how to prove
that a function on relations is an up-to technique and whether it can be safely
used in a passive clause or not.

\subsection{Up-to techniques, general definitions}%
\label{ss:lamcal-up-to-nf-general}

We recall here the definitions from our previous
work~\cite{Aristizabal-al:FSCD16}. The goal of up-to techniques is to simplify
bisimulation proofs: instead of proving that a relation~$\rel$ is a
bisimulation, we show that $\rel$ respects some looser constraints which still
imply bisimilarity. In our setting, we distinguish the up-to techniques which
can be used in passive clauses (called \emph{strong} up-to techniques), from the
ones which cannot.

\begin{defi}%
  \label{def:up-to}
  An up-to technique (resp.\ strong up-to technique) is a function $f$ on
  relations such that $\rel \pprogress \rel, f(\rel)$ (resp.\ $\rel \pprogress
  f(\rel), f(\rel)$) implies $\rel \mathop\subseteq \nfbisim$. When $f$ is an
  up-to technique (resp.\ strong up-to technique) and $\rel \pprogress \rel,
  f(\rel)$ (resp.\ $\rel \pprogress f(\rel), f(\rel)$), we say that $\rel$ is a
  bisimulation up to $f$.
\end{defi}

\begin{exa}%
  \label{ex:upto-refl}
  In Example~\ref{ex:wads}, we have to include identical terms (related by
  $\idbis$) in the definition of the candidate relation, although we already
  know such terms are bisimilar. To avoid doing so, we define the up to
  reflexivity technique $\rawrefl$ as follows.
  \begin{mathpar}
    \inferrule{ }{t \utrefl\rel t}
  \end{mathpar}
  In our setting, we show that $\rawrefl$ is a strong up-to technique and can be
  used after active and passive clause. We can then define $\rel$ in
  Example~\ref{ex:wads} without $\idbis$ and show it is a bisimulation up to
  reflexivity.
\end{exa}

\begin{exa}%
  \label{ex:upto-red-ctx}
  In deterministic languages, a common up-to technique is to reason \emph{up to
    reduction}~\cite{Sangiorgi-Pous:11}, which allows terms to reduce before
  being related.
  \begin{mathpar}
    \inferrule{\redrtc t {t'} \\ \redrtc s {s'} \\ t' \rel s'} {t \utred\rel s}
  \end{mathpar}
  With this technique, one can ignore the intermediary reduction steps and
  reason in a big-step way even with a small-step semantics. We prove in this
  section that $\rawutred$ is another example of strong up-to technique. In
  contrast, up to context, discussed in Section~\ref{ss:lamcal-up-to-nf} is not
  strong and cannot be used after passive transitions.
\end{exa}

Proving that a given~$f$ is an up-to technique is difficult with
Definition~\ref{def:up-to}, in part because this notion is not stable under
composition or union~\cite{Madiot:PhD}. Following Madiot, Pous, and
Sangiorgi~\cite{Sangiorgi-Pous:11,Madiot-al:CONCUR14}, we rely on a notion
of~\emph{compatibility}, which gives sufficient conditions for~$f$ to be an
up-to technique, and is easier to establish, as functions built out of
compatible functions using composition and union remain compatible.

We first need some auxiliary notions on notations on functions on relations,
ranged over by $f$, $g$, and $h$ in what follows. We define $f \subseteq g$ and
$f \cup g$ argument-wise, e.g., $(f \cup g)(\rel)=f(\rel) \cup g(\rel)$ for all
$\rel$. We define $f^\omega$ as $\bigcup_{n \in \mathbb N} f^n$. We write
$\rawid$ for the identity function on relations, and $\fid f$ for
$f \mathop\cup \rawid$. The technique $\rawrefl$ of Example~\ref{ex:upto-refl}
differs from $\rawid$ as the former builds the identify on terms (for all $\rel$
and $t$, $t \utrefl\rel t$), while the latter is the identity on relations (for
all $\rel$, $t$, and $s$, $t \rel s$ implies $t \utid\rel s$).

Given a set $\setF$ of functions, we also write~$\setF$ for the function defined
as $\bigcup_{f \in \setF} f$. We say a function $f$ is \emph{generated from
  $\setF$} if $f$ can be built from functions in $\setF$ and $\rawid$ using
union, composition, and $\cdot^\omega$. The largest function generated from
$\setF$ is ${\fid \setF}^\omega$.

A function~$f$ is \emph{monotone} if $\rel \mathop\subseteq \rels$ implies
$f(\rel) \mathop\subseteq f(\rels)$.  We write $\finpower\rel$ for the set of
finite subsets of $\rel$, and we say $f$ is \emph{continuous} if it can be
defined by its image on these finite subsets, i.e., if
$f(\rel) \mathop\subseteq \bigcup_{\rels \in \finpower\rel}f(\rels)$. The up-to
techniques of the present paper are defined by inference rules with a finite
number of premises, like $\rawutred$ or $\rawrefl$, so we can easily show they
are continuous. Continuous functions are interesting because of their
properties:\footnote{Our formalization revealed an error in previous
  works~\cite{Aristizabal-al:FSCD16,Madiot:PhD} which use $f$ instead of
  $\fid f$ in the last property of Lemma~\ref{l:continuity} (expressing
  idempotence of ${\fid f}^\omega$)---$\rawid$ has to be factored in for the
  property to hold.}

\begin{lem}%
  \label{l:continuity}
  If $f$ and $g$ are continuous, then $f \comp g$ and $f \cup g$ are continuous.
  If $f$ is continuous, then $f$ is monotone, and $f \comp {\fid f}^\omega
  \subseteq {\fid f}^\omega$.
\end{lem}

Compatibility relies on a notion of \emph{evolution} for functions on relations,
which can be seen has the higher-order equivalent of progress.
\begin{defi}
  A function $f$ evolves to $g, h$, written $f \fevolve g, h$, if for all $\rel
  \pprogress \rel, \relt$, we have $f(\rel) \pprogress g(\rel), h(\relt)$. A
  function $f$ \emph{strongly} evolves to $g, h$, written $f \sevolve g, h$, if
  for all $\rel \pprogress \rels, \relt$, we have $f(\rel) \pprogress g(\rels),
  h(\relt)$.
\end{defi}
\noindent
We have $f \fevolve g, h$ or $f \sevolve g, h$ if terms related by $f(\rel)$
become related by respectively $g(\rel)$ and $h(\relt)$, or $g(\rels)$ and
$h(\relt)$, depending on how $\rel$ progresses. Like with progress, evolution
distinguishes passive from active clauses, so that $g$ is used after passive
clauses and $h$ after active ones. We further distinguish between regular
evolution $\fevolve$ and strong evolution $\sevolve$: for regular evolution, we
consider $\rel$ such that $\rel \pprogress \rel, \relt$, while strong evolution
is more liberal, as it allows for relations $\rel$ such that $\rel \pprogress
\rels, \relt$. It follows the discrepancy between regular and strong up-to
techniques, the former being not allowed after passive clauses, while the latter
are.

The next examples illustrate strong evolution; for an example of regular
evolution, see Lemma~\ref{l:app-lambda}.

\begin{exa}
  It is easy to prove that $\rawrefl \sevolve \rawrefl, \rawrefl$, as identical
  terms remain equal when they progress.
\end{exa}

\begin{exa}%
  \label{ex:red-sevolve}
  We show that
  $\rawutred \sevolve \rawutred \cup \rawid, \rawutred \cup \rawid$. Let
  $\rel \pprogress \rels, \relt$. We first have to show the inclusions
  $\utred\rel \mathord{\subseteq} \utred\rels \cup \rels$ and
  $\utred\rels \cup \rels \mathord{\subseteq} \utred\relt \cup \relt$; these
  hold because $\rel \mathord{\subseteq} \rels$,
  $\rels \mathord{\subseteq} \relt$ (by definition of $\pprogress$) and
  $\rawutred$ is continuous and therefore monotone.

  Next, let $t \utred \rel s$; by definition, there exist $t'$ and $s'$ such
  that $\redrtc t {t'}$, $\redrtc s {s'}$, and $t' \rel s'$. We consider in turn
  the different progress clauses of Definition~\ref{def:progress}.

  For the reduction clause, let $t''$ such that $\red t {t''}$; because this
  clause is active, we want to find $s''$ such that $\redrtc s {s''}$ and
  $t'' \mathrel{(\rawutred \cup \rawid)(\relt)} s''$. We distinguish two cases:
  first, suppose $t \neq t'$. Because the semantics is deterministic, we have
  $\redrtc {t''}{t'}$, and therefore $t'' \utred \rel s$ holds. But
  $\rel {\mathord{\subseteq}} \relt$ by definition of progress, and $\rawutred$
  is monotone, hence we have $t'' \utred \relt s$, which implies
  $t'' \mathrel{(\rawutred \cup \rawid)(\relt)} s$, as wished.

  Suppose now that $t = t'$; then $t \rel s'$. Because $\rel \progress \rels,
  \relt$ and $\red t {t''}$, there exists $s''$ such that $\redrtc {s'}{s''}$
  and $t'' \relt s''$. Therefore, $\redrtc s {s''}$, and $t'' \relt s''$ implies
  $t'' \mathrel{(\rawutred \cup \rawid)(\relt)} s''$, as required.

  For the passive value clause, we have $t = v$, and we want to find $w$ such
  that $\redrtc s w$ and $v \testval \rels w$. Because $t$ cannot reduce, we
  have $t'=v$, and therefore $v \rel s'$. Because $\rel \pprogress \rels,
  \relt$, there exists $w$ such that $\redrtc {s'} w$ and $v \testval \rels
  w$. But we also have $\redrtc s w$, hence the result holds. The reasoning for
  the active open stuck term clause is the same, but with~$\relt$ instead of
  $\rels$.
\end{exa}

Like in Madiot et al.~\cite{Madiot-al:CONCUR14}, we say that a set of functions
$\setF$ is compatible if each function~$f$ in $\setF$ evolves towards functions
generated from $\setF$. However, in contrast with Madiot et al., we need some
restrictions on the resulting combinations, which depend on whether $f$ is a
strong up-to technique or not.

\begin{defi}
  \label{def:diac-comp}
  A set $\setF$ of continuous functions is \emph{diacritically compatible} if
  $\setF$ has a subset~$\setS$ such that
  \begin{itemize}
  \item for all $f \in \setS$, we have $f \sevolve {\fid\setS}^\omega,
    {\fid\setF}^\omega$;
  \item for all $f \in \setF$, we have $f \fevolve {\fid\setS}^\omega \comp
    \fid\setF \comp {\fid\setS}^\omega, {\fid\setF}^\omega$.
  \end{itemize}
\end{defi}

\noindent
The (possibly empty) subset $\setS$ represents the strong up-to techniques of
$\setF$, and a function $f$ may evolve differently depending on whether $f$
belongs to $\setS$ or not. In both cases, $f$ may evolve towards any function
generated from $\setF$ after an active clause. The difference is after passive
clauses, where a function in $\setS$ has to evolve towards a function generated
from $\setS$ only, while composing with a non-strong function is allowed if
$f \notin \setS$. In the latter case, we progress from $f(\rel)$ with $f$ not
strong, and we expect to progress towards a combination which still
includes~$f$. It is safe to do so, as long as $f$ (or in fact, any non-strong
function in $\setF$) is used at most once.

If $\setS_1$ and~$\setS_2$ are subsets of~$\setF$ which verify the conditions of
Definition~\ref{def:diac-comp}, then $\setS_1 \cup \setS_2$ also does, so there
exists the largest subset of $\setF$ which satisfies the conditions, written
$\strong \setF$. The next lemma shows that being in a compatible set is a
sufficient criterion to be a (possibly strong) up-to technique. In practice,
proving that $f$ is in a compatible set~$\setF$ is easier than proving directly
it is an up-to technique.

\begin{lem}
  \label{l:properties-compatibility-better}
  Let $\setF$ be a diacritically compatible set.
  \begin{itemize}
  \item If
    $\rel \pprogress {\widehat{\mathsf{strong}(\setF)}}^\omega(\rel),
    {\fid\setF}^\omega(\rel)$, then ${\fid\setF}^\omega(\rel)$ is a
    bisimulation.
  \item any function generated from $\setF$ is an up-to technique, and any
    function generated from $\strong\setF$ is a strong up-to technique.
  \item For all $f \in \setF$, we have $f(\bisim) \mathop\subseteq \bisim$.
  \end{itemize}
\end{lem}

\noindent The second point implies that combining functions from a compatible
set using union, composition, or $\cdot^\omega$ produces up-to techniques. In
particular, if $f \in \setF$, then $f$ is an up-to technique, and similarly, if
$f \in \strong\setF$, then $f$ is a strong up-to technique. The last item states
that bisimilarity respects compatible functions, so proving that up to context
is compatible implies that bisimilarity is preserved by contexts.

The first item suggests a more flexible notion of up-to technique, as it shows
that given a compatible set $\setF$, a relation may progress towards different
functions $f$ and $g$, $\rel \pprogress f(\rel), g(\rel)$, and still be included
in the bisimilarity as long as $f$ is generated from $\strong\setF$ and $g$ is
generated from $\setF$. In what follows, we rely on that property in examples
and say that in that case, $\rel$ is a bisimulation up to $\setF$, or $\rel$ is
a bisimulation up to $f_1 \ldots f_n$ if $\setF = \{ f_1 \ldots f_n \}$.

\begin{exa}
  The set $\setF \defeq \{ \rawutred \}$ is diacritically compatible, with
  $\strong \setF = \{ \rawutred \}$. Indeed, we prove in
  Example~\ref{ex:red-sevolve} that $\rawutred$ strongly progresses towards
  combinations of $\rawutred$ and $\rawid$ which respects the conditions of
  Definition~\ref{def:diac-comp}. As a consequence, $\rawutred$ is a strong
  up-to technique by Lemma~\ref{l:properties-compatibility-better}, and the
  bisimilarity is preserved by reduction: if $\redrtc t {t'}$, $\redrtc s {s'}$,
  and $t' \bisim s'$, then $t \bisim s$.

  Similarly, $\{ \rawrefl \}$ is diacritically compatible and $\rawrefl$ is
  strong, and $\utrefl \bisim \mathop\subseteq \bisim$ implies that two equal
  terms are bisimilar.
\end{exa}

\begin{exa}%
  \label{ex:wads-refl-red}
  We can simplify the definition of $\rel$ in Example~\ref{ex:wads} to just
  \[\mathord{\rel} \defeq \{ (\lam x x, \wads), (y, \lam x {y \iapp (\app \wads
      x)}) \mid y \mbox{ fresh} \}\] and show that $\rel$ is a bisimulation up
  to $\rawrefl$ and $\rawutred$.
\end{exa}

\subsection{Up to context for normal-form bisimilarity}%
\label{ss:lamcal-up-to-nf}

Our primary goal in this section is to prove that reasoning up to context is an
up-to technique. We let $C$ range over contexts, i.e., terms with a
hole~$\mtevctx$, and define reasoning up to context as follows.
\begin{mathpar}
  \inferrule{t \rel s}
  {\inctx C t \utctx\rel \inctx C s}
\end{mathpar}
Unlike $\rawutred$ or $\rawrefl$ in the previous section, we cannot prove up to
context is an up-to technique by itself. We need some extra techniques, but we
also decompose $\rawutctx$ into smaller techniques, to allow for a finer-grained
distinction between strong and regular up-to techniques.

Figure~\ref{fig:upto-lambda} presents the techniques we consider for the
$\lambda$-calculus. The substitutive closure $\rawutsubst$ is already used in
previous works~\cite{Lassen:MFPS99,Lassen:LICS06,Biernacki-al:HAL15}. The
closure by evaluation contexts $\rawutectx$ is more unconventional, although we
define it in a previous work~\cite{Biernacki-al:HAL15}. It is not the same as
reasoning up to context, since we can factor out different contexts, as long as
they are related when we plug a fresh variable inside them. It is reminiscent of
$\star$-bisimilarity~\cite{Aristizabal-al:FSCD16} which can also factor out
different contexts in its up-to techniques, except that $\star$-bisimilarity
compares contexts with values and not simply variables.

We first show that we can indeed build $\rawutctx$ out of the techniques of
Figure~\ref{fig:upto-lambda}. Closure w.r.t. $\lambda$-abstraction is achieved
through $\rawlam$, and closure w.r.t.\ variables is a consequence of $\rawrefl$,
as we have $x \utrefl \rel x$ for all $x$. We can construct an application out
of $\rawutectx$ and $\rawrefl$.
\begin{lem}%
  \label{l:app}
  If $t \rel t'$ and $s \rel s'$, then $\app t s \mathrel{(\rawutectx \comp
    (\rawid \cup \rawutectx \comp (\rawid \cup \rawrefl)))(\rel)} \app
  {t'}{s'}$.
\end{lem}
\noindent Let $x$ be a fresh variable; then $\app x \hole \testevctx{\utrefl
  \rel} \app x \hole$. Combined with $s \rel s'$, it implies $\app x s \utectx
          {(\rawid \cup \rawrefl)(\rel)} \app x {s'}$, i.e., $\app \hole s
          \testevctx{\utectx {(\rawid \cup \rawrefl)(\rel)}} \app \hole
                    {s'}$. This combined with $t \rel t'$ using $\rawutectx$
                    gives the result of Lemma~\ref{l:app}. In the end, if we
                    define $\rawapp$ as the combination of techniques of
                    Lemma~\ref{l:app}, then the following holds.

\begin{lem}%
  \label{l:utctx-carac}
  $t \utctx\rel s$ iff $t \mathrel{{(\rawrefl \cup \rawlam \cup \rawapp)}^\omega}
  s$.
\end{lem}

Proving that $\setF \defeq \{ \rawrefl, \rawlam, \rawutsubst, \rawutectx,
\rawid, \rawutred\}$ is diacritically compatible amounts to showing that each
function in $\setF$ evolves towards functions generated from~$\setF$ which
respect Definition~\ref{def:diac-comp}. Before discussing some of the
compatibility proofs, we compare ourselves to Lassen's
proof~\cite{Lassen:MFPS99}. Lassen defines a closure combining all the
techniques of Figure~\ref{fig:upto-lambda}, and then proves by induction on its
definition that it is a bisimulation, using Definition~\ref{def:nf-classic}.

A first difference is that our proofs are in a small-step style while Lassen's
is big-step, which means that he has to perform an extra induction on the number
of reduction steps. Our technique is also more modular, as we can isolate
smaller techniques and do each compatibility proof separately, instead of
reasoning on a single closure. Apart from these (minor) points, our proof and
Lassen's are quite similar, as they consist in case analyses on the possible
reductions the related terms can make. Therefore, the main difference between
our setting and Lassen's is not so much the proof technique itself, but the
notion of progress it is based upon. The crucial case is when proving
Lemma~\ref{l:app-lambda}, discussed below, where we can conclude thanks to
diacritical progress, while we do not know how to complete the proof with a
regular notion of progress (Definition~\ref{def:nf-classic}).

\medskip

\begin{figure}
\begin{mathpar}
\inferrule{ }{t \utrefl\rel t}
\and
\inferrule{t \rel s}{\lam x t \utlam\rel \lam x s}
\and
\inferrule{t \rel s \\  v \testval\rel w }
{\subst t x v \utsubst\rel \subst s x w}
\\
\inferrule{t \rel s \\ E \testevctx\rel E'}
{\inctx E t \utectx\rel \inctx {E'} {s}}
\and
\inferrule{\redrtc t {t'} \\ \redrtc s {s'} \\ t' \rel s'}
{t \utred\rel s}
\end{mathpar}
\caption{Up-to techniques for the $\lambda$-calculus}%
\label{fig:upto-lambda}
\end{figure}

We now sketch some evolution proofs, starting with the simplest ones. We already
discussed the proofs for $\rawutred$ and $\rawrefl$ in
Section~\ref{ss:lamcal-up-to-nf-general}. The strong evolution proof
for~$\rawlam$ is also quite simple.

\begin{lem}
  $\rawlam \sevolve \rawlam \cup \rawutred, \rawlam \cup \rawutred$.
\end{lem}

\begin{proof}[Sketch]
Let $\rel \pprogress \rels, \relt$; we want to prove that
$\utlam\rel \pprogress \utlam\rels \cup \utred\rels, \utlam\relt \cup
\utred\relt$.  Let $\lam x t \utlam\rel \lam x s$ such that $t \rel s$. The only
clause to check is the one for values: we have $(\lam x t) \iapp x \rawred t$
and $(\lam x s) \iapp x \rawred s$, i.e., $(\lam x t) \iapp x \utred\rel (\lam x
s) \iapp x$, which implies $(\lam x t) \iapp x \utred\rels (\lam x s) \iapp x$
because $\rel \mathord{\subseteq} \rels$ and $\rawutred$ is monotone.
\end{proof}

The technique $\rawutsubst$ is also strong, as we can show the following result.

\begin{lem}%
  \label{lem:subst}
  $\rawutsubst \sevolve \rawutsubst, (\rawid \cup \rawutectx) \comp \rawutsubst
  \comp (\rawid \cup \rawutsubst)$.
\end{lem}

\begin{proof}[Sketch]
Let $\rel \pprogress \rels, \relt$, and $\subst t x v
\utsubst\rel \subst s x w$ such that $t \rel s$ and $v \testval\rel w$. We check
the different clauses by case analysis on $t$. If $t$ is a value, then there
exists a value $s'$ such that $\redto s {s'}$ and $t \testval\rel s'$. But then
$\subst t x v$ and $\subst {s'} x w$ are also values, and then we can prove that
$\subst t x v \testval{\utsubst\rels} \subst {s'} x w$ holds. If $t \rawred t'$,
then there exists $s'$ such that $\redrtc s {s'}$ and $t' \relt s'$. Then
$\subst t x v \rawred \subst {t'} x v$, $\redrtc {\subst s x w} {\subst {s'} x
  w}$, and $\subst {t'} x v \utsubst \relt \subst {s'} x w$.

Finally, if $t=\inctx E {y \iapp v'}$, then there exists $s'$ such that
$\redto s {s'}$ and $t \testopen\relt s'$. If $y \neq x$, then $\subst t x v$
and $\subst {s'} x w$ are open stuck terms in
$\testopen{\utsubst\relt}$. Otherwise, we distinguish cases based on whether $v$
is a $\lambda$-abstraction or not. In the former case, let $v = \lam z {t'}$,
$s' = \inctx {E'}{x \iapp w'}$. Then
$\subst t x v = \inctx {\subst E x v}{v \iapp \subst {v'} x v} \rawred \inctx
{\subst E x v}{\subst {t'} z {\subst {v'} x v} }$.  From $v \testval\rel w$ and
$v \iapp z \rawred t'$, we know that there exists $s''$ such that
$\redrtc {w \iapp z}{s''}$ and $t' \relt s''$.  Consequently, we have
$\redrtc{{\redrtc {\subst s x w}{\subst {s'} x w}}={\inctx {\subst {E'} x w}{w
      \iapp \subst{w'} x w }}}{\inctx {\subst {E'} x w}{\subst {s''} z
    {\subst{w'} x w}}}$.  Then
$\inctx {\subst E x v}{x'} \utsubst\relt \inctx {\subst E x w}{x'}$ for a fresh
$x'$, but also
$\subst {t'} z {\subst {v'} x v} \utsubst{\utsubst \relt} \subst {s''} z
{\subst{w'} x w}$, so after plugging, we obtain terms in
$\rawutectx \comp \rawutsubst \comp (\rawid \cup \rawutsubst)(\relt) $.

If $v$ is a variable, a similar reasoning shows that $\subst t x v$ and
$\subst {s'} x w$ evaluate to open stuck terms, whose contexts are related by
$\rawutectx \comp \rawutsubst \comp (\rawid \cup \rawutsubst)(\relt) $ and whose
arguments are related by $\utsubst{\utsubst \relt}$.
\end{proof}

The proof for $\rawutsubst$ does not require the clause for values to be
passive, and is thus similar to the corresponding subcase in Lassen's proof by
induction~\cite{Lassen:MFPS99}. In contrast, we need testing values to be
passive when dealing with $\rawutectx$; we present the problematic subcase
below. This is where our proof technique differs from Lassen's, as we do not
know how to make this subcase go through in a proof by induction.

\begin{lem}%
  \label{l:app-lambda}
  $\rawutectx \evolve \rawutectx, \rawutsubst \cup \rawutectx \cup
  (\rawid \cup \rawutectx) \comp \rawutsubst \comp (\rawid \cup \rawutsubst)$.
\end{lem}

\begin{proof}[Sketch]
Let $\rel \pprogress \rel, \rels$, and $\inctx E t \utectx\rel
\inctx {E'} s$ such that $E \testevctx\rel E'$ and $t \rel s$.  We proceed by
case analysis on $E$ and $t$. Most cases are straightforward; the problematic
case is when~$t$ is a variable $x$ and $E = \app \hole w$. Because $t \rel s$,
there exists $v_2$ such that $\redto{s}{v_2}$ and $x \testval\rel v_2$. Because
$E \testevctx\rel E'$, we have $\inctx E y \rel \inctx {E'} y$ for a fresh $y$,
and therefore $\inctx E x \utsubst\rel \inctx {E'}{v_2}$. We can conclude using
Lemma~\ref{lem:subst}: there exists an open stuck term~$s'$ such that
$\redrtc{\redrtc{\inctx{E'}{s}}{\inctx{E'}{v_2}}}{s'}$ and $\app x w
\testopen{((\rawid \cup \rawutectx) \comp \rawutsubst \comp (\rawid \cup
  \rawutsubst)(\rels))} s'$.

In an induction proof with Definition~\ref{def:nf-classic}, we would have in
that case $x \testval\rels v_2$ instead of $\rel$, and $\redto {\inctx {E'}
  y}{s'}$ for some $s'$ such that $\app y w \testopen\rels s'$. We do not see
how to go further in the case $v_2$ is a $\lambda$-abstraction: we have to prove
that $\subst {s'} y {v_2}$ evaluates to an open stuck term, but we do not have
any progress hypothesis about $\rels$.
\end{proof}

Each technique of $\setF$ evolves towards functions generated from $\setF$
respecting Definition~\ref{def:diac-comp}: for passive clauses, $\rawutred$,
$\rawrefl$, $\rawlam$, and $\rawutsubst$ strongly evolves towards combinations
of strong techniques, while $\rawutectx$ evolves towards itself, thus respecting
the criterion of at most one not strong technique in that case. As a result, the
following theorem holds.

\begin{thm}%
  \label{t:compatible-lambda}
  The set $\setF$ is diacritically compatible, with $\strong \setF = \setF
  \setminus \{ \rawutectx \}$.
\end{thm}

A direct consequence of Lemma~\ref{l:properties-compatibility-better} is that
any function generated from $\setF$ is an up-to technique, so in particular
$\rawutctx$ (cf
Lemma~\ref{l:utctx-carac}). Lemma~\ref{l:properties-compatibility-better} then
also implies that~$\nfbisim$ is preserved by contexts. As it is easy to show
that $\nfbisim$ is also an equivalence relation, we have the following
corollary.

\begin{cor}
  $\bisim$ is a congruence.
\end{cor}

\noindent
This corollary, in turn, immediately implies the soundness of $\bisim$
w.r.t.\ the usual contextual equivalence of the $\lambda$-calculus, where we
observe termination of evaluation~\cite{Abramsky-Ong:IaC93}---the notion of
contextual equivalence that we take throughout the paper. However, as proved
in~\cite{Lassen:LICS05},~$\bisim$ is not complete w.r.t.\ contextual equivalence.

\begin{exa}%
  \label{ex:wads-upto}
  With the same definitions as in Examples~\ref{ex:wads}, let
  \[ v \defeq \fix{\lam {zxy}{\app z x}} \mbox{ and } w \defeq \fix{\lam {zxy}{z
        \iapp {(\wads \iapp x)}}}.\]
  We prove these values are bisimilar by showing that
  \begin{align*}
    \mathord{\rel} & \defeq \{ (v, w), (v \iapp x, w \iapp x), (x, \lam y {x \iapp (\app
           \wads y)}) \mid x \mbox{ fresh} \}
  \end{align*}
  is a bisimulation up to $\setF$. For the first pair, we have directly $v
  \testval\rel w$. For the second pair, we have
  \[
    \app v x \rawredrtc (\lam {zxy}{\app z x}) \iapp v \iapp x \rawredrtc \lam y
    {v \iapp x} \mbox{ and } \app w x \rawredrtc (\lam {zxy}{z \iapp {(\wads
        \iapp x)}}) \iapp w \iapp x \rawredrtc \lam y {w \iapp (\wads \iapp x)}.
  \]
  We show the two resulting values are related up to $\setF$. Indeed, we have
  $v \rel w$, and $x \utred\rel \app \wads x$, because
  $\wads \iapp x \rawredrtc \lam y {x \iapp (\app \wads y)}$, therefore
  $\lam y {v \iapp x} \utlam{\utapp{(\rawid \cup \rawutred)(\rel)}} \lam y {w
    \iapp (\wads \iapp x)}$. Note that we can use $\rawapp$---which is built out
  of the not strong technique $\rawutectx$ (cf. Lemma~\ref{l:app})---in that
  case, as the reduction clause is active.

  For the last pair, we compare $x \iapp z$ and $\lam y {x \iapp (\app \wads y)}
  \iapp z$, but $\lam y {x \iapp (\app \wads y)} \iapp z \rawredrtc x \iapp \lam
  y {z \iapp (\app \wads y)}$, so we can conclude with up to reduction and
  reflexivity.
\end{exa}

\section{Delimited-control operators}%
\label{s:delcon}

We show that the results of Section~\ref{s:lamcal} seamlessly carry over to the
call-by-value $\lambda$-calculus extended with \textshift{} and
\textreset{}~\cite{Danvy-Filinski:LFP90,Biernacki-al:HAL15}, thus demonstrating
the robustness of the approach, but also improving on the previous results on
extensional normal-form bisimulations for this
calculus~\cite{Biernacki-al:HAL15}.

\subsection{Syntax, semantics, and normal-form bisimulations}%
\label{ss:delcon-calculus}

We extend the grammar of terms and values given in Section~\ref{s:lamcal} as
follows:

\[
\begin{array}{rcl}
  t,s & \bnfdef & \dots \bnfor \reset{t}
  \\
  v,w & \bnfdef & \dots \bnfor \rawshift
\end{array}
\]

\vspace{2mm}\noindent where $\rawreset$ is the control delimiter \textreset{}
and $\rawshift$ is the delimited-control operator \textshift{}. Usually,
\textshift{} is presented as a binder
$\shift{x}{t}$~\cite{Danvy-Filinski:LFP90,Biernacki-al:HAL15} or as a special
form $\app{\rawshift}{t}$~\cite{Filinski:POPL94}, but here we choose a more
liberal syntax treating \textshift{} as a value (as, e.g.,
in~\cite{Kameyama:HOSC07}). This makes the calculus a little more interesting
since \textshift{} becomes a subject to $\eta$-expansion just as any other
value, and moreover it makes it possible to study terms such as
$\app{\rawshift}{\rawshift}$. We call \emph{pure terms} effect-free terms, i.e.,
values and terms of the form $\reset{t}$.

We distinguish a subclass of pure contexts ($E$) among evaluation contexts
($F$):

\[
\begin{array}{rcl}
  \pctx
  & \bnfdef &
  \mtpctx
  \bnfor \valpctx v \pctx
  \bnfor \argpctx \pctx t
  \\
  \evctx
  & \bnfdef &
  \mtevctx
  \bnfor \valevctx v \evctx
  \bnfor \argevctx \evctx t
  \bnfor \resevctx \evctx
\end{array}
\]

\vspace{2mm}\noindent We extend the function $\mathsf{fv}$ to both kinds of
contexts. Note that an evaluation context~$F$ is either pure or can be written
$\inctx{F'}{\reset{E'}}$ for some $F'$ and $E'$. Pure contexts can be captured
by $\rawshift$, as we can see in the following rules defining the call-by-value
left-to-right reduction semantics of the calculus:

\begin{align*}
  \inctx \evctx {\app {(\lam x t)} v} & \rawred
  \inctx \evctx {\subst t x v}
  \\[1mm]
  \inctx \evctx {\reset{\inctx \pctx {\app{\rawshift}{v}}}} & \rawred
  \inctx \evctx {\reset{\app{v}{\lam x {\reset {\inctx \pctx x}}}}}
  \mbox{ with } x \not\in \fv \pctx
  \\[1mm]
  \inctx \evctx {\reset v} & \rawred
  \inctx \evctx v
\end{align*}

\vspace{2mm}\noindent The first rule is the usual call-by-value
$\beta$-reduction. When $\rawshift$ is applied to a value~$v$, it captures its
surrounding pure context $\pctx$ up to the dynamically nearest enclosing reset,
and provides its term representation $\lam x {\reset {\inctx \pctx x}}$ as an
argument to $v$. Finally, a reset which encloses a value can be removed, since
the delimited subcomputation is finished.  All these reductions may occur within
a metalevel context~$\evctx$ that encodes the chosen call-by-value evaluation
strategy. As in Section~\ref{s:lamcal}, the reduction relation~$\rawred$ is
preserved by evaluation contexts.

\begin{exaC}[\cite{Biernacki-al:HAL15}]%
  \label{e:reduction}
  Let $i \defeq \lam x x$, $\omega \defeq \lam x {\app x x}$, and $\Omega \defeq \app
  \omega \omega$. We present the sequence of reductions initiated by
  $\reset {\app {(\app {(\app {\rawshift}
        {\lam{k}{\app{i}{(\app{k}{i})}}})} {(\app {\rawshift}
        {\lam{k}{\omega}})})} {\Omega}}$:

  \[
     \begin{array}{rcl}
       \reset
           {\app
             {(\app
               {(\app
                 {\rawshift}
                 {\lam{k}{\app{i}{(\app{k}{i})}}})}
               {(\app
                 {\rawshift}
                 {\lam{k}{\omega}})})}
             {\Omega}} & \rawred & \quad\quad\quad (1)
           \\[1mm]
           \reset
               {\app
                 {(\lam{k}{\app{i}{(\app{k}{i})}})}
                 {(\lam{x}
                   {\reset
                     {\app
                       {\app
                         {x}
                         {(\app
                           {\rawshift}
                           {\lam{k}{\omega}})}}
                       {\Omega}}})}} & \rawred & \quad\quad\quad (2)
               \\[1mm]
               \reset
                   {\app
                     {i}
                     {(\app
                       {(\lam{x}
                         {\reset
                           {\app
                             {\app
                               {x}
                               {(\app
                                 {\rawshift}
                                 {\lam{k}{\omega}})}}
                             {\Omega}}})}
                       {i})}} & \rawred & \quad\quad\quad (3)
                   \\[1mm]
               \reset
                   {\app
                     {i}
                     {\reset
                        {\app
                          {\app
                            {i}
                            {(\app
                              {\rawshift}
                              {\lam{k}{\omega}})}}
                          {\Omega}}}} & \rawred & \quad\quad\quad (4)
                   \\[1mm]
               \reset
                   {\app
                     {i}
                     {\reset
                       {\app
                         {(\lam{k}{\omega})}
                         {(\lam{x}
                           {\reset
                             {\app
                               {\app
                                 {i}
                                 {x}}
                               {\Omega}}})}}}} & \rawred & \quad\quad\quad (5)
                   \\[1mm]
               \reset
                   {\app
                     {i}
                     {\reset
                       {\omega}}} & \rawred & \quad\quad\quad (6)
                   \\[1mm]
               \reset
                   {\app
                     {i}
                     {\omega}} & \rawred & \quad\quad\quad (7)
                   \\[1mm]
               \reset
                   {\omega} & \rawred & \quad\quad\quad (8)
                   \\[1mm]
               \omega & &
     \end{array}
     \]

     \vspace{2mm}\noindent In step (1) the pure context $\argpctx{(\argpctx
       \hole {(\app {\rawshift} {\lam k \omega})})} {(\app \omega \omega)}$, is
     captured and reified as a term that in step (2) is substituted for $k$ in
     the argument of \textshift{}. In step (3) the captured context is
     reactivated by $\beta$-reduction, and thanks to the \textreset{} enclosing
     the body of the lambda representing the captured context, it is not
     \emph{merged} with the current context, but \emph{composed} with it. As a
     result, the capture triggered by \textshift{} in step (4) leaves the outer
     $i$ intact.~\footnote{For this reason \textshift{} and \textreset{} are
       called \textit{static} control operators, in contrast to Felleisen's
       \textcontrol{} and \textprompt{}~\cite{Felleisen:POPL88} that are called
                   \textit{dynamic} control operators~\cite{Biernacki-al:SCP06}.}
     In step (5) the context captured in step (4) is discarded, again in terms
     of $\beta$-reduction. In steps (6) and (8) a control delimiter guarding a
     value is removed, whereas in step (7) a regular function application takes
     place.  Note that even though the reduction strategy is call-by-value, some
     function arguments are not evaluated, like the non-terminating term
     $\Omega$ in this example.
\end{exaC}

\begin{exa}
  This example illustrates the operational behavior of $\rawshift$ as a value:
  \vspace{1mm}
  \begin{align*}
    \reset{\inctx{E}{\app{\rawshift}{\rawshift}}} & \rawred \\
    \reset{\app{\rawshift}{\lam{x}{\reset{\inctx{E}{x}}}}} & \rawred \\
    \reset{\app{(\lam{x}{\reset{\inctx{E}{x}}})}{(\lam{x}{\reset{x}})}} & \rawred \\
    \reset{\reset{\inctx{E}{\lam{x}{\reset{x}}}}} &
  \end{align*}

  \vspace{2mm}\noindent In particular, if $E = \mtpctx$, then the value of the
  initial term, after two additional reduction steps, is $\lam{x}{\reset{x}}$,
  i.e., the representation of the empty context in the calculus of delimited
  control.
\end{exa}

A term $t$ either uniquely reduces to another term, or is an eager normal form:
it is either a value $v$, an open stuck term $\inctx{F}{\app{x}{v}}$, or a \textit{
  control-stuck term} $\inctx{\pctx}{\app{\rawshift}{v}}$. The latter cannot
reduce further since it lacks a reset enclosing~$\rawshift$. In the original
reduction semantics~\cite{Biernacka-al:LMCS05}, derived from the higher-order
evaluator implementing the denotational semantics of \textshift and
\textreset~\cite{Danvy-Filinski:LFP90}, it was assumed that programs are
evaluated in an enclosing top-level \textreset{}. Here, however, we consider a
relaxed semantics~\cite{Biernacki-al:HAL15} that lifts this requirement. Such a
semantics corresponds to some of the existing implementations of
delimited-control operators~\cite{Filinski:POPL94} that make it possible to
observe control-stuck programs (raising a ``missing reset'' exception if one
forgot the top-level reset) and it scales to control operators that can remove
the enclosing delimiter, e.g., as in general calculi with multiple
prompts~\cite{Dybvig-al:JFP06,Aristizabal-al:LMCS17}. This choice does not
influence the operational semantics of \textshift and \textreset in any other
way.

Because $\textshift{}$ can decompose contexts, we have to change the relation
$\testevctx{R}$ as discussed in~\cite{Biernacki-al:HAL15}:

\begin{mathpar}
  \inferrule
      {\inctx E x \rel \inctx {E'} x \\ x \mbox{ fresh}}
      {E \testevctx\rel E'}
  \and
  \inferrule
      {\reset{\inctx E x} \rel \reset{\inctx {E'} x} \\
        \inctx{F}{x} \rel \inctx{F'}{x} \\
        x \mbox{ fresh}}
      {\inctx{F}{\reset{E}} \testevctx\rel {\inctx{F'}{\reset{E'}}}}
\end{mathpar}

\vspace{2mm}\noindent We also introduce a relation $\testctrl{R}$ to handle
control-stuck terms:

\begin{mathpar}
  \inferrule
      {E \testevctx\rel E' \\
       \reset{\app v x} \rel \reset{\app w x} \\
       x \mbox{ fresh}}
      {\inctx E {\app \rawshift v} \testctrl\rel \inctx {E'}{\app \rawshift {w}}}
\end{mathpar}

\vspace{2mm}\noindent whereas the relation $\testval\rel$ remains unchanged, so
that it accounts for the $\eta$-law, even though the values now include
$\rawshift$.

We can now define (extensional) normal-form bisimulation and bisimilarity for
the extended calculus, again using the notion of diacritical progress.
\begin{defi}%
  \label{def:nf-delcont}
  A relation $\rel$ diacritically progresses to $\rels$, $\relt$ written $\rel
  \pprogress \rels, \relt$, if $\rel \mathop\subseteq \rels$, $\rels
  \mathop\subseteq \relt$, and $t \rel s$ implies:
  \begin{itemize}
  \item if $\red t {t'}$, then there exists $s'$ such that $\redrtc s {s'}$ and
    $t' \relt s'$;
  \item if $t=v$, then there exists $w$ such that $\redto s w$, and $v
    \testval\rels w$;
  \item if $t=\inctx F {\app x v}$, then there exist $F'$, $w$ such that $\redto
    s {\inctx {F'}{\app x w}}$ and $\inctx F {\app x v} \testopen\relt \inctx
    {F'}{\app x w}$;
  \item if $t=\inctx E {\app \rawshift v}$, then there exist $E'$, $w$ such that
    $\redto s {\inctx {E'}{\app \rawshift {w}}}$ and $\inctx E {\app \rawshift
    v} \testctrl\relt \inctx {E'}{\app \rawshift w}$;
  \item the converse of the above conditions on $s$.
  \end{itemize}
  A normal-form bisimulation is a relation $\rel$ such that $\rel \pprogress
  \rel, \rel$, and normal-form bisimilarity $\nfbisim$ is the union of all
  normal-form bisimulations.
\end{defi}
\noindent Only the clause for values is passive, as in
Definition~\ref{def:progress}.

\begin{exa}%
\label{ex:double-shift}
The terms $\app{\rawshift}{\rawshift}$ and
$\app{\rawshift}{(\lam{k}{\app{k}{(\lam{x}{x})}})}$ are bisimilar since the
following relation is a normal-form bisimulation:

\[
  \begin{array}{ll@{\hspace{5mm}}ll}
    \idbis \cup  \{  & \hspace{-0.5em} (\app{\rawshift}{\rawshift}, &
    \app{\rawshift}{(\lam{k}{\app{k}{(\lam{x}{x})}})}), &\quad (1) \\
     & \hspace{-0.5em} (\reset{\app{\rawshift}{z}}, &
    \reset{\app{(\lam{k}{\app{k}{(\lam{x}{x})}})}{z}}), &\quad (2) \\
     & \hspace{-0.5em} (\reset{\app{z}{(\lam{x}{\reset{x}})}}, &
    \reset{\app{z}{(\lam{x}{x})}}), &\quad (3)\\
     & \hspace{-0.5em} (\app{(\lam{x}{\reset{x}})}{y}, & \app{(\lam{x}{x})}{y}), & \quad(4) \\
     & \hspace{-0.5em} (\reset{y}, & y) \;\} & \quad(5)
  \end{array}
  \]

  \vspace{2mm}\noindent where $\idbis$ is the identity relation and $x$, $y$,
  and $z$ fresh variables. In $(1)$ we compare two control-stuck terms, so to
  validate the bisimulation conditions, we have to compare the two empty
  contexts (which are in $\testevctx\idbis$) and the arguments of $\textshift$.
  Here, extensionality plays an important role, as these arguments are of
  different kinds ($\rawshift$ vs a $\lambda$-abstraction). We compare them by
  passing them a fresh variable $z$, thus we include the pair $(2)$ in the
  bisimulation. The terms of $(2)$ can be reduced to those in $(3)$, where we
  compare open stuck terms, so we have to include $(4)$ to compare the arguments
  of $z$. The terms in $(4)$ can then be reduced to the ones in $(5)$ which in
  turn reduce to identical terms.
\end{exa}

\subsection{Up-to techniques}%
\label{ss:delcon-up-to-nf}

\begin{figure}
\begin{mathpar}
\inferrule{t \rel s \\ E \testevctx\rel E'}
{\inctx E t \utpctx\rel \inctx {E'} {s}}
\and
\inferrule{t \rel s \\ \reset{E} \testevctx\rel \reset{E'}}
{\reset{\inctx E t} \utpctxrst\rel \reset{\inctx {E'} s}}
\and
\inferrule{t \rel s  \\ t, s \mbox{ pure} \\ \inctx F x \rel \inctx {F'} x \\ x \mbox{ fresh}}
{\inctx F {t} \utectxpure\rel \inctx {F'}{s}}
\end{mathpar}
\caption{Up-to techniques specific to the $\lambda$-calculus extended with
  \textshift{} and \textreset{}}%
\label{fig:upto-shift}
\end{figure}

The up-to techniques we consider for this calculus are the same as in
Figure~\ref{fig:upto-lambda}, except we replace $\rawutectx$ by three
more fine-grained up-to techniques defined in
Figure~\ref{fig:upto-shift}. The techniques $\rawpctx$, $\rawpctxrst$
allow to factor out related pure contexts and pure contexts with a
surrounding reset. The third one ($\rawectxpure$) can be used only
with pure terms, but uses a naive comparison between any evaluation
contexts instead of~$\testevctx\cdot$.  Indeed, a pure term cannot
evaluate to a control-stuck term, so decomposing contexts
with~$\testevctx\cdot$ is not necessary. The usual reasoning up to
evaluation context $\rawutectx$ can be obtained by composing these
three up-to techniques.

\begin{lem}
If $t \rel t'$ and $F \testevctx\rel F'$ then
$\inctx F t \mathrel{(\rawpctx \cup (\rawectxpure\comp\rawpctxrst))(\rel)} \inctx{F'}{t'}$.
\end{lem}

\noindent We do not define extra up-to techniques corresponding to the new
constructs of the language, as $\textshift$ is dealt with like variables---using
$\rawrefl$, and closure w.r.t. $\textreset$ can be deduced from $\rawpctxrst$ by
taking the empty context. Defining a dedicated up-to technique for $\textreset$
would have some merit since it could be proved strong. It is not so for
$\rawpctxrst$, as we can see in the next theorem:

\begin{thm}%
  \label{t:compatible-shift}
  The set $\setF \defeq \{ \rawrefl, \rawlam, \rawutsubst,
  \rawpctx, \rawpctxrst, \rawectxpure, \rawid, \rawutred\}$ is
  diacritically compatible, with $\strong \setF = \setF \setminus \{
  \rawpctx, \rawpctxrst, \rawectxpure \}$.
\end{thm}

\noindent
The evolution proofs are by case analysis on the possible reductions that the
related terms can do. The techniques $\rawpctx$, $\rawpctxrst$, and
$\rawectxpure$ are not strong as they exhibit the same problematic case
presented in Lemma~\ref{l:app-lambda} (for $\rawpctx$ and $\rawectxpure$) or a
slight variant ($\reset \pctx = \reset{\app \hole v}$ and $t =x$ for
$\rawpctxrst$).

The main difference with the $\lambda$-calculus evolution proofs is the case
analysis on the behavior of \textshift, which may perform a capture or not,
depending on its surrounding evaluation context. For example, in the case of
$\rawutsubst$, the substituted value may be \textshift, which, when replacing
$x$ in a term $t = \inctx {F_1}{\app x v}$, may either produce a control-stuck
term if~$F$ is pure, or otherwise reduce. The hypotheses we have on the
evaluation contexts thanks to $\testevctx\cdot$ allow us to conclude in each
case. Indeed, if $t \rel s$ and $\rel \pprogress \rels,\relt$, then
$\redto s {s'}$ for some $s' = \inctx {F_2}{\app x w}$ such that
$F_1 \testevctx\relt F_2$ and $v \testval\relt w$. If $F_1$ is pure, then $F_2$
is also pure and $\subst t x \rawshift$ and $\subst {s'} x \rawshift$ are both
control-stuck. Otherwise, $F_1 = \inctx {F'_1}{\reset {E_1}}$ and
$F_2 = \inctx {F'_2}{\reset {E_2}}$ for some $F'_1$, $F'_2$, $E_1$, and $E_2$
such that $\inctx {F'_1} y \relt \inctx {F'_2} y$ and
$\reset{\inctx {E_1} y} \relt \reset{\inctx {E_2} y}$ for a fresh $y$. Then
$\red {\subst t x \rawshift}{\subst{\inctx {F'_1}{\reset{\app v {\lam y
          {\reset{\inctx {E_1} y}}}}}} x \rawshift}$ and
$\red{\subst {s'} x \rawshift}{\subst{\inctx {F'_2}{\reset{\app w {\lam y
          {\reset{\inctx {E_2} y}}}}}} x \rawshift}$, and the two resulting
terms can be related using mainly $\rawectxpure$, $\rawlam$ and $\rawutsubst$.

As in Section~\ref{s:lamcal}, from
Theorem~\ref{t:compatible-shift} and
Lemma~\ref{l:properties-compatibility-better} it follows that $\bisim$
is a congruence, and, therefore, is sound w.r.t.\ the contextual
equivalence of~\cite{Biernacki-al:HAL15}; it is not complete as showed
in op.\ cit.

\begin{exa}
  With these up-to techniques, we can simplify the bisimulation of
  Example~\ref{ex:double-shift} to just a single pair
  \[\mathord{\rel} \defeq \{ (\app{\rawshift}{\rawshift},
  \app{\rawshift}{(\lam{k}{\app{k}{(\lam{x}{x})}})}) \}.\]
  Indeed, we have $\lam x {\reset x} \utlam{\utred{\utrefl \rel}} \lam x x$, and
  $z \iapp \hole \testevctx{\utrefl \rel} z \iapp \hole$, so (3) is in
  $f(\rel)$, where
  $f \defeq \rawpctxrst \comp (\rawrefl \cup (\rawlam \comp \rawutred \comp
  \rawrefl))$.
  Then, (2) is in $\utred{f(\rel)}$, and for~(1), we have also to relate $\hole$
  with $\hole$, thus (1) is in $\utrefl\rel \cup \utred{f(\rel)}$. As a result,
  $\rel$ is a bisimulation up to $\rawutred$, $\rawrefl$, $\rawpctxrst$, and
  $\rawlam$.
\end{exa}

\section{Abortive Control}%
\label{s:abortcon}

In contrast with delimited-control operators, abortive control
operators capture the whole surrounding context and reify it as an
abortive computation. As such they are considerably more challenging
as far as modular reasoning is concerned in that one has to take into
account reduction of complete programs, rather than of their
subexpressions. In this section we apply our technique to a calculus
with \textcallcc and
\textabort~\cite{Felleisen-Friedman:FDPC3,Felleisen-Hieb:TCS92}, for
which no sound normal-form bisimilarity has been defined so far.

\subsection{Syntax and semantics}%
\label{ssec:syntax-lac}

The syntax of the $\lambda$-calculus with abortive control operators
\textcallcc ($\callcc$) and \textabort ($\abrt$) can be defined as the
following extension of the syntax of~Section~\ref{s:lamcal}:

\[
\begin{array}{rcl}
  t, s & \bnfdef & \ldots \bnfor \abort t
  \\
  v,w & \bnfdef & \ldots \bnfor \callcc
\end{array}
\]

\vspace{2mm}\noindent with the same syntax of evaluation contexts as
for the $\lambda$-calculus:

\[
  E \bnfdef \hole \bnfor \app v E \bnfor \app E t
\]

\vspace{2mm}\noindent However, the reduction rules for abortive
control operators considered in this section, unlike the ones for the
pure lambda calculus or delimited control operators, are not
compatible with respect to evaluation contexts, i.e., $t \rawred t'$
does not imply $\inctx E t \rawred \inctx E {t'}$. As such they are
meant to describe evaluation of complete programs. In order to reflect
this requirement, we introduce a separate syntactic category of
programs $p$, and we adjust the grammar of terms accordingly:

\[
  \begin{array}{rcl}
    t, s & \bnfdef & \ldots \bnfor \abort p
    \\
    p, q & \bnfdef & t
  \end{array}
\]

\vspace{2mm}\noindent The remaining productions for terms as well as
the grammars of values and evaluation contexts do not change. However,
we introduce a category of \emph{program contexts} $F$, comprising all
evaluation contexts

\[
F \bnfdef E
\]

\vspace{2mm}\noindent and make it explicit in the reduction rules
below that they give semantics to complete programs by manipulating
program contexts. While this modification may seem cosmetic, it
reflects the idea of abortive continuations representing ``the rest of
the computation'' and it allows us to introduce, later on in this
section, a minimal extension that makes it possible to define sound
normal-form bisimulations for the original calculus.

The call-by-value semantics of the calculus is defined by the
following rules.
\vspace{1mm}
\begin{align*}
  \inctx F {\app {(\lam x t)} v} & \rawred
  \inctx F {\subst t x v}
  \\[1mm]
  \inctx F {\app \callcc v} & \rawred \inctx F {\app v {\lam y
                                       {\abort {\inctx F y}}}} \mbox{ with } y
                                       \notin \fv F
  \\[1mm]
  \inctx F {\abort p} & \rawred p
\end{align*}

\vspace{2mm}\noindent When the capture operator $\callcc$ is applied
to a value $v$, it makes a copy of the current program context $F$ as
a value and passes it to $v$. The captured context $F$ is reified as a
$\lambda$-abstraction whose body is guarded by \textabort, which when
evaluated removes its enclosing context, effectively restoring $F$ as
the current program context. Hence the syntactic requirement that
\textabort be applied to a program and not to an arbitrary term. The
semantics of the control operators is illustrated in the following
examples.

\begin{exa}
  Assuming that $x \not\in \fv{F_2} \cup \fv{v}$, we have the
  following reduction sequence:

  \begin{align*}
    \inctx{F_1}{\app{\callcc}{\lam{x}{\inctx{F_2}{\app{x}{v}}}}} & \rawred \\
    \inctx{F_1}{\app{(\lam{x}{\inctx{F_2}{\app{x}{v}}})}{\lam{y}{\abort{\inctx{F_1}{y}}}}} & \rawred \\
    \inctx{F_1}{\inctx{F_2}{\app{\lam{y}{\abort{\inctx{F_1}{y}}}}{v}}} & \rawred \\
    \inctx{F_1}{\inctx{F_2}{\abort{\inctx{F_1}{v}}}} & \rawred \\
    \inctx{F_1}{v} &
  \end{align*}
\end{exa}

\begin{exa}
  This example illustrates the operational behavior of \textcallcc as
  a value:

  \begin{align*}
    \inctx{F}{\app{\callcc}{\callcc}} & \rawred \\
    \inctx{F}{\app{\callcc}{\lam{x}{\abort{\inctx{F}{x}}}}} & \rawred \\
    \inctx{F}{\app{(\lam{x}{\abort{\inctx{F}{x}}})}{\lam{x}{\abort{\inctx{F}{x}}}}} & \rawred \\
    \inctx{F}{\abort{\inctx{F}{\lam{x}{\abort{\inctx{F}{x}}}}}} & \rawred \\
    \inctx{F}{\lam{x}{\abort{\inctx F {x}}}} &
  \end{align*}

  \vspace{2mm}\noindent
  In particular, if $F = \mtectx$, then the value of the initial term
  is $\lam{x}{\abort{x}}$, i.e., the representation of the empty
  context in the calculus of abortive control.
\end{exa}

\subsection{Normal-form bisimulation}

In the presence of abortive control operators the notion of normal
form makes sense only for complete programs. Therefore, in order to
define a notion of normal-form bisimulation for the calculus under
consideration we need to somehow take into account how two given terms
behave in any program context. In order to make such tests possible we
introduce a distinct set of \emph{context variables}, ranged over by
$\cv$, to represent abstract contexts that are a dual concept to the
fresh variables used for testing functional values, as argued in
Remark~\ref{r:abstrvars}.

We only need to extend the syntax of programs and program contexts:

\[
\begin{array}{rcl}
  p & \bnfdef & \dots \bnfor \appk \cv t
  \\[1mm]
  F & \bnfdef & \dots \bnfor \appk \cv E
\end{array}
\]

\vspace{2mm}\noindent A program $\appk{\cv}{t}$ stands for the term
$t$ plugged in a program context represented by the context variable
$\cv$---we deliberately abuse the notation. Similarly, a program
context $\appk \cv E$ stands for the context $E$ completed with a
program context represented by $\cv$. We write $\fk p$ for the set of
context variables of a program $p$.

Free variables represent unknown or abstract values in a program and
the substitution $\subst p x v$ can be seen as an operation that
replaces an abstract value $x$ with a concrete value $v$.  We define
an analogous concretization operation for context variables that we
call \emph{context substitution}: program $\substc p \cv F$ is a
program $p$, where all occurrences of $\cv$ are replaced by $F$.  More
formally, we have

\[
\begin{array}{rcl}
\substc {\appk \cv t} \cv F & = & \inctx F {\substc t \cv F} \\[1mm]
\substc {\appk \cw t} \cv F & = & \appk \cw {\substc t \cv F}
\qquad \textrm{when $\cw \neq \cv$} 
\end{array}
\]

\vspace{2mm}\noindent and for terms the context substitution is
applied recursively to its sub-terms and sub-programs. Note that in
the first case the substitution plugs the term $\substc t \cv F$ in
the program context $F$. Context substitution preserves reduction.

\begin{lem}%
  \label{l:red-ab}
  If $\red p q$, then $\red{\substc p \cv F}{\substc q \cv F}$.
\end{lem}

\noindent
When we substitute for a context variable, this property can be seen
as a form of preservation of reduction by program contexts.  For
example, if $p = \appk \cv {(\app \callcc \lam x t)}$ with
$\cv \notin \fk F$, then
$\redrtc p {\appk \cv {\subst t x {\lam y {\abort {\appk \cv y}}}}}$,
and
$\redrtc {\substc p \cv F}{\inctx F {\subst t x {\lam y {\abort
        {\inctx F y}}}}}$.

Normal forms of the extended language are either values $v$, open
stuck terms $\inctx F {\app x v}$, or \emph{context-stuck terms} of
the form $\appk \cv v$. We define progress, bisimulation, and
bisimilarity on programs of the extended language, and we define
$\testtm \rel$ to apply these notions to terms. We change $\testval
\rel$ accordingly and we remind the definitions of $\testevctx \rel$
and $\testopen \rel$.

\begin{mathpar}
  \inferrule{\appk \cv t \rel \appk \cv s \\ \cv \mbox{ fresh}}
  {t \testtm\rel s}
  \and
  \inferrule{\app v x \testtm\rel \app  w x \\ x \mbox{ fresh}}
  {v \testval\rel w}
  \\
  \inferrule{\inctx F x \rel \inctx {F'} x \\ x \mbox{ fresh}}
  {F \testevctx\rel F'}
  \and
  \inferrule{F \testevctx\rel F' \\ v \testval\rel w}
  {\inctx F {\app x v} \testopen\rel \inctx{F'}{\app x w}}
\end{mathpar}

\vspace{2mm}
\begin{defi}
  \label{def:progress-ab}
  A relation $\rel$ \textit{diacritically progresses} to $\rels$, $\relt$
  written $\rel \pprogress \rels, \relt$, if $\rel \mathop\subseteq
  \rels$, $\rels \mathop\subseteq \relt$, and $p \rel q$ implies:
  \begin{itemize}
  \item if $\red p {p'}$, then there exists $q'$ such that $\redrtc q {q'}$ and $p'
    \relt q'$;
  \item if $p=v$, then there exists $w$ such that $\redto q w$;
  \item if $p=\inctx F {\app x v}$, then there exist $F'$, $w$ such that
    $\redto q {\inctx {F'}{\app x w}}$ and $\inctx F {\app x v}
    \testopen\relt \inctx {F'}{\app x w}$;
  \item if $p=\appk \cv v$, then there exist $w$ such that
    $\redto q \appk \cv w$ and $v \testval\rels w$;
  \item the converse of the above conditions on $q$.
  \end{itemize}
  A normal-form bisimulation is a relation $\rel$ such that
  $\rel \pprogress \rel, \rel$, and normal-form bisimilarity $\nfbisim$ is the
  union of all normal-form bisimulations.
\end{defi}

In the value case $p = v$, we simply ask $q$ to evaluate to $w$ without
requiring anything of $v$ and $w$; we therefore equate all values when
considered as programs. A value without context variable implies that the
context has been aborted, and therefore $v$ or $w$ cannot be applied to a
testing variable. As an example, $\abort v$ and $\abort w$ are bisimilar for all
$v$ and $w$; when plugged into a context $F$, these terms remove their current
evaluation context as soon as they are run, so there is no way for the context
to pass an argument to these values.

In contrast, if $p=\appk \cv v$, then we require that $q$ evaluates to a
context-stuck term $\appk \cv w$ with the same context variable $\cv$, and we
test the two values by applying them to an abstract value in an abstract
context. Indeed, the program context represented by $\cv$ can either discard the
value plugged in it or it can apply it to an argument. Obviously, we cannot
equate $\appk \cv v$ and $\appk {\cw} v$ if $\cv \neq \cw$, as $\cv$ and $\cw$
stand for possibly different program contexts.

The bisimilarity presented here has a similar structure to the
bisimilarity for $\lambda\mu$-calculus~\cite{Parigot:LPAR92} presented
in~\cite{Stoevring-Lassen:POPL07}. It should come as no surprise,
because the programs and program contexts of the extended calculus are
reminiscent of, respectively, named terms and named contexts of the
$\lambda\mu$-calculus, and the context substitution coincides with
some presentations of structural substitution in
$\lambda\mu$-calculus~\cite{Ariola-al:HOSC07,Stoevring-Lassen:POPL07}.
The resemblance is superficial though. In particular, context
variables come from the `abstract context' interpretation and are not
bound by any construct. The structure of programs is also richer: not
all programs have to be considered in an abstract context, so the
second case of the definition~\ref{def:progress-ab} is characteristic
of the calculus with \textcallcc{} and \textabort{}, and is not
present in the definition of bisimulations for $\lambda\mu$-calculus.

\begin{rem}%
  \label{r:pgms-in-gen}
  One could introduce the notion of program and program context along
  with context variables in Section~\ref{s:lamcal} and~\ref{s:delcon}
  to define the notion of normal-form bisimulation explicitly on
  programs of the form $\appk{\cv}{t}$, in the spirit of
  Definition~\ref{def:progress-ab}. However, since the constructs in
  these calculi do not manipulate the complete program contexts, such
  definitions would trivially reduce to the ones we have presented.
\end{rem}

We use the \textcallcc axiomatization of Sabry and
Felleisen~\cite{Sabry-Felleisen:LFP92} as a source of examples.

\begin{exa}
  To prove that
  $\lam x {\app \callcc {\lam y {\app x y}}} \testval\nfbisim \callcc$
  ($\eta_{v_2}$ axiom), we compare
  $\appk \cv {\app {(\lam x {\app \callcc {\lam y {\app x y}}})} z}$ and
  $\appk \cv {\app \callcc z}$, and both reduce to
  $\appk \cv {\app z {\lam x {\abort {\appk c x}}}}$.
\end{exa}

\begin{exa}%
  \label{ex:current}
  The $C_{current}$ axiom equates $\app \callcc {\lam x {\app x t}}$ and
  $\app \callcc {\lam x t}$ for all $t$. We have to relate
  $\appk \cv {\app \callcc {\lam x {\app x t}}}$ and
  $\appk \cv {\app \callcc {\lam x t}}$ for a fresh $\cv$. Reducing these
  programs give us respectively
  $\appk \cv {\app {(\lam y {\abort {\appk \cv y}})}{\subst t x {\lam y {\abort
          {\appk \cv y}}}}}$ and
  $\appk \cv {\subst t x {\lam y {\abort {\appk \cv y}}}}$. From there, we can
  conclude with a tedious case analysis on the behavior of $t$; we need to
  distinguishes cases based on whether
  $\appk \cw {\subst t x {\lam y {\abort {\appk \cv y}}}}$ (where $\cw$ is
  fresh) evaluates to a value, an open stuck term, or an in-context value, with
  $\cw$ as a context variable or not. The proof is greatly simplified with up-to
  techniques (see Example~\ref{ex:current-upto}).
\end{exa}

\begin{figure}
\begin{mathpar}
\inferrule{ }{p \utrefl\rel p}
\and
\inferrule{ }{v \utresult\rel w}
\and
\inferrule{p \rel q}{\appk \cv {\abort p} \utabort\rel \appk \cv {\abort q}}
\and
\inferrule{t \testtm\rel s}{\appk \cv {\lam x t} \utlam\rel \appk \cv {\lam x s}}
\and
\inferrule{p \rel q \\ v \testval\rel w}
{\subst p x v \utsubstv\rel \subst q x w}
\and
\inferrule{p \rel q \\ F \testevctx\rel F' }
{\substc p \cv F \utsubstc\rel \substc q \cv {F'}}
\and
\inferrule{\redrtc p {p'} \\ \redrtc {q} {q'} \\ p' \rel q'}{p \utred\rel q}
\end{mathpar}
\caption{Up-to techniques for the $\lambda$-calculus with \textcallcc and \textabort}%
\label{fig:upto-callcc}
\end{figure}

\subsection{Up-to techniques} Figure~\ref{fig:upto-callcc} presents the up-to
techniques for the $\lambda$-calculus with \textcallcc and \textabort. The main
difference with the previous calculi is that reasoning up to evaluation
contexts $\rawutectx$ is replaced by $\rawsubstc$, which is bit more general. We
can deduce the former from the latter, as $t \testtm\rel s$ implies
$\appk \cv t \rel \appk \cv s$ for a fresh $\cv$, which in turn implies
$\substc {\appk \cv t} \cv F \utsubstc\rel \substc {\appk \cv s} \cv {F'}$,
i.e., $\inctx F t \utsubstc\rel \inctx {F'} s$. But we can also factorize
several occurrences of a context with $\rawsubstc$, as, e.g.,
$\inctx F {\abort {\inctx F t}}$ can be written
$\substc {\inctx \cv {\abort {\appk \cv t}}} \cv F$ for $\cv \notin \fk t$.

The other novelty is $\rawresult$, which expresses the fact that
values do not need to be tested when they are not in-context. Among
all these up-to techniques, only $\rawsubstc$ is not strong, which is
expected as it behaves like $\rawutectx$.

\begin{thm}%
  \label{t:compatible-calcc}
  The set $\setF \defeq \{ \rawrefl, \rawresult,
  \rawabort, \rawlam, \rawsubstv, \rawsubstc, \rawutred \}$ is
  diacritically compatible, with $\strong \setF = \setF \setminus \{ \rawsubstc \}$.
\end{thm}
Like for the previous calculi, we show that each function in $\setF$
evolves towards a combination of functions in $\setF$, by doing a case
analysis on the possible reductions of the related terms. Context
variables make the treatment of the capture by \textcallcc a bit
different than for \textshift, as now each capture should be written
as a context substitution, to be able to conclude using
$\rawsubstc$. For example, consider as in
Section~\ref{ss:delcon-up-to-nf} $\subst t x \callcc \utsubstv\rel
\subst s x \callcc$ for some $t = \inctx F {\app x v}$ and $\rel$ such
that $\rel \pprogress \rels, \relt$. There exists $F'$ and $w$ such
that $\redto s {\inctx {F'}{\app x w}}$, $F \testevctx\relt F'$ and $v
\testval\relt w$. Then $\red {\subst t x \callcc}{\subst {\inctx F
    {\app v {\lam y {\abort {\inctx F y}}}}} x \callcc}$ and $\redrtc
              {\subst s x \callcc}{\subst {\inctx {F'}{\app w {\lam y
                      {\abort {\inctx {F'} y}}}}} x \callcc}$ for a
              fresh $y$; to relate the two resulting terms, we start
              from $v \testval\relt w$, which by definition implies
              $\appk \cv {\app v z} \relt \appk \cv {\app w z}$ for
              fresh $\cv$ and $z$. Then because $\lam y {\abort {\appk
                  \cv y}} \utrefl \relt \lam y {\abort {\appk \cv
                  y}}$, we can relate $\appk \cv {\app v \lam y
                {\abort {\appk \cv y}}}$ and $\appk \cv {\app w \lam y
                {\abort {\appk \cv y}}}$ with $\rawrefl$ and
              $\rawsubstv$, and then $\subst{\substc {\appk \cv {\app
                    v \lam y {\abort {\appk \cv y}}}} \cv F} x
              \callcc$ and $\subst{\substc {\appk \cv {\app w \lam y
                    {\abort {\appk \cv y}}}} \cv {F'}} x \callcc$
              using also $\rawsubstc$. The last two programs are equal
              to the ones we want to relate, so we can conclude.

As before, we get as a corollary of Theorem~\ref{t:compatible-calcc}
and Lemma~\ref{l:properties-compatibility-better} that
$\testtm\nfbisim$ is a congruence. The bisimilarity is also sound
w.r.t.\ contextual equivalence when we restrict ourselves to the
calculus of Section~\ref{ssec:syntax-lac}. (It can be shown that it is
not complete, e.g., $\lam{x}{\app{x}{I}}$ and
$\lam{x}{\app{(\lam{y}{\app{x}{I}})}{(\app{x}{I})}}$ are contextually
equivalent, but not bisimilar.)
\begin{thm}%
  \label{t:soundness-callcc}
  Let $t$ and $s$ be terms built from the grammar of
  Section~\ref{ssec:syntax-lac}. If $t \testtm\nfbisim s$, then~$t$ and~$s$ are
  contextually equivalent.
\end{thm}
\begin{proof}
  Let $C$ be a context built from the grammar of Section~\ref{ssec:syntax-lac}
  which closes $t$ and $s$. By Theorem~\ref{t:compatible-calcc} and
  Lemma~\ref{l:properties-compatibility-better}, we have
  $\inctx C t \nfbisim \inctx C s$.  Therefore, by the definition of the
  bisimilarity, $\inctx C t$ terminates iff $\inctx C s$ terminates.
\end{proof}

\begin{exa}%
  \label{ex:current-upto}
  We prove the $C_{current}$ axiom (Example~\ref{ex:current}) by showing that
  \[\rel \defeq \{(\appk \cv {\app \callcc {\lam x {\app x t}}}, \appk \cv {\app
      \callcc {\lam x t}})\}\] is a bisimulation up to $\rawsubstc$,
  $\rawutred$, and $\rawrefl$. These programs reduce in two steps to
  respectively
  $\appk \cv {\app {(\lam y {\abort {\appk \cv y}})}{\subst t x {\lam y {\abort
          {\appk \cv y}}}}}$ and
  $\appk \cv {\subst t x {\lam y {\abort {\appk \cv y}}}}$. These two programs
  can be written $\substc p \cw {F_1}$ and $\substc p \cw {F_2}$ with $\cw$
  fresh, $p = \appk \cw {\subst t x {\lam y {\abort {\appk \cv y}}}}$,
  $F_1 = \appk \cv {\app{(\lam y {\abort {\appk c y}})} \hole}$, and
  $F_2 = \appk \cv \hole$. Because
  $\redrtc{\appk \cv {\app{(\lam y {\abort {\appk c y}})} z}}{\appk \cv z}$ for
  a fresh $z$, we have $F_1 \testevctx{\utred{\utrefl \rel}} F_2$, which implies
  \(\appk \cv {\app {(\lam y {\abort {\appk \cv y}})}{\subst t x {\lam y {\abort
            {\appk \cv y}}}}} \utsubstc{\utred{\utrefl \rel}}\appk \cv {\subst t x
      {\lam y {\abort {\appk \cv y}}}}. \)
\end{exa}

\begin{exa}
  The $C_{tail}$ axiom relates $\app \callcc {\lam y {\app {(\lam x t)} s}}$ and
  $\app {(\lam x {\app \callcc {\lam y t}})} s$ if $y \notin \fv s$. When
  surrounded with a fresh context variable $\cv$, the former reduces in two
  steps to $\appk \cv {\app {(\lam x {\subst t y {\lam z {\abort {\appk \cv
              z}}}})} s}$, so we prove that
  \[\rel \defeq \{(\appk \cv {\app {(\lam x {\subst t y {\lam z {\abort {\appk \cv
                z}}}})} s}, \appk \cv {\app {(\lam x {\app \callcc {\lam y t}})}
      s})\}\] is a bisimulation up to $\rawsubstc$, $\rawutred$, and
  $\rawrefl$. Again, we notice that these programs can be written
  $\substc {\appk \cw s} \cw {F_1}$ and $\substc {\appk \cw s} \cw {F_2}$ with
  $\cw$ fresh,
  $F_1 = \appk \cv {\app {(\lam x {\subst t y {\lam z {\abort {\appk \cv z}}}})}
    \hole}$, and
  $F_2 = \appk \cv {\app {(\lam x {\app \callcc {\lam y t}})} \hole}$. If $z'$
  is a fresh variable, then
  $\redrtc{\inctx {F_1}{z'}}{\appk \cv {\subst{\subst t y {\lam z {\abort {\appk
              \cv z}}}} x {z'}}}$ and
  $\redrtc{\inctx {F_2}{z'}}{\appk \cv {\subst {\subst t x {z'}} y {\lam z
        {\abort {\appk \cv z}}}}}$. The variables being pairwise distinct, the
  order of substitutions does not matter, therefore $\inctx {F_1}{z'}$ and
  $\inctx {F_2}{z'}$ reduce to identical programs. We can then conclude as in
  Example~\ref{ex:current-upto}. Without up-to techniques, we would have to do a
  case analysis on the behavior of $s$ to conclude.
\end{exa}
The remaining axioms~\cite{Sabry-Felleisen:LFP92} can be proved easily in a
similar way thanks to up-to techniques.

\section{Conclusion}%
\label{s:conclusion}

In this article we present a new approach to proving soundness of
normal-form bisimilarities as well as of bisimulations up to context
that allow for $\eta$-expansion. The method we develop is based on our
framework~\cite{Aristizabal-al:FSCD16} that generalizes the work of
Madiot et al.~\cite{Madiot-al:CONCUR14,Madiot:PhD} in that it allows
for a special treatment of some of the clauses in the definition of
bisimulation. In particular, we show soundness of an extensional
bisimilarity for the call-by-value $\lambda$-calculus and of the
corresponding up to context, where it is critical that
comparing values in a way that respects $\eta$-expansion is done
passively, i.e., by requiring progress of a relation to
itself. Following the same route, we obtained similar results for the
extension of the call-by-value $\lambda$-calculus with delimited
control, where the set of normal forms is richer, and we believe this
provides an evidence for the robustness of the method. To the best of
our knowledge, there has been no soundness proof of extensional
normal-form up to context technique for any of the two calculi
before. Furthermore, we successfully applied our technique to a theory
of extensional normal-form bisimulations for the call-by-value
$\lambda$-calculus with \textcallcc and \textabort, which has not existed in the
literature before.

The proof method we propose should trivially apply to the existing
non-ex\-ten\-sio\-nal whnf
bisimilarities~\cite{Lassen:99,Lassen:MFPS99,Lassen:MFPS05} and extensional hnf
bisimilarities~\cite{Lassen:MFPS99,Lassen:LICS06} for the $\lambda$-calculus and
its variants. Since such whnf bisimilarities do not take into account
$\eta$-expansion, their testing of values would be active and all their up-to
techniques would be strong, so actually Madiot's original framework is
sufficient to account for them. In the case of extensional hnf bisimilarities,
normal forms are generated by the grammar:

\[
\begin{array}{rcl}
  h & \bnfdef & \lam{x}{h} \bnfor n
  \\
  n & \bnfdef & x \bnfor \app{n}{t}
\end{array}
\]

\vspace{2mm}\noindent and in order to account for $\eta$-expansion a
$\lambda$-abstraction $\lam{x}{h}$ is related to a normal form~$n$, provided $h$
is related to $\app{n}{x}$, a freshly created normal form. Then, relating normal
forms $\app{\app{x}{t_1}}{\dots t_m}$ and $\app{\app{y}{s_1}}{\dots s_n}$
requires that $x = y$, $m = n$, and for all $i$, $t_i$ is related to
$s_i$. Thus, in extensional hnf bisimulations the relation on normal forms
provides enough information to make testing of normal forms active just like in
non-extensional whnf bisimulations.

A possible direction for future research is to investigate whether our
method can be adapted to enriched normal-form bisimulations that take
advantage of an additional structure akin to environments used in
environmental bisimulations~\cite{Sangiorgi-al:TOPLAS11}. One such
theory is the complete enf bisimilarity for the
$\lambda\mu\rho$-calculus of sequential control and
state~\cite{Stoevring-Lassen:POPL07}, for which the existing proof of
congruence is intricately involved and no up-to context technique has
been developed.

\subsubsection*{Acknowledgments} We would like to thank the anonymous
reviewers of MFPS 2017 and LMCS for their helpful comments on the
presentation of this work.

\bibliographystyle{abbrv}
\bibliography{lmcs-mfps}

\begin{thebibliography}{10}

\bibitem{Abramsky-Ong:IaC93}
S.~Abramsky and C.-H.~L. Ong.
\newblock Full abstraction in the lazy lambda calculus.
\newblock {\em Information and Computation}, 105:159--267, 1993.

\bibitem{Ariola-al:HOSC07}
Z.~M. Ariola, H.~Herbelin, and A.~Sabry.
\newblock A proof-theoretic foundation of abortive continuations.
\newblock {\em Higher-Order and Symbolic Computation}, 20(4):403--429, 2007.

\bibitem{Aristizabal-al:FSCD16}
A.~Aristiz{\'a}bal, D.~Biernacki, S.~Lenglet, and P.~Polesiuk.
\newblock {Environmental Bisimulations for Delimited-Control Operators with
  Dynamic Prompt Generation}.
\newblock In D.~Kesner and B.~Pientka, editors, {\em 1st International
  Conference on Formal Structures for Computation and Deduction (FSCD 2016)},
  volume~52 of {\em Leibniz International Proceedings in Informatics (LIPIcs)},
  pages 9:1--9:17, Dagstuhl, Germany, 2016. Schloss Dagstuhl--Leibniz-Zentrum
  f{\"u}r Informatik.

\bibitem{Aristizabal-al:LMCS17}
A.~Aristiz{\'a}bal, D.~Biernacki, S.~Lenglet, and P.~Polesiuk.
\newblock {Environmental Bisimulations for Delimited-Control Operators with
  Dynamic Prompt Generation}.
\newblock {\em Logical Methods in Computer Science}, 13(3), 2017.

\bibitem{Barendregt:84}
H.~Barendregt.
\newblock {\em The Lambda Calculus: Its Syntax and Semantics}, volume 103 of
  {\em Studies in Logic and the Foundation of Mathematics}.
\newblock North-Holland, revised edition, 1984.

\bibitem{Biernacka-al:LMCS05}
M.~Biernacka, D.~Biernacki, and O.~Danvy.
\newblock An operational foundation for delimited continuations in the {CPS}
  hierarchy.
\newblock {\em Logical Methods in Computer Science}, 1(2:5):1--39, Nov. 2005.

\bibitem{Biernacki-al:SCP06}
D.~Biernacki, O.~Danvy, and C.~Shan.
\newblock On the static and dynamic extents of delimited continuations.
\newblock {\em Science of Computer Programming}, 60(3):274--297, 2006.

\bibitem{Biernacki-Lenglet:FLOPS12}
D.~Biernacki and S.~Lenglet.
\newblock Normal form bisimulations for delimited-control operators.
\newblock In T.~Schrijvers and P.~Thiemann, editors, {\em FLOPS'12}, number
  7294 in LNCS, pages 47--61, Kobe, Japan, May 2012. Springer-Verlag.

\bibitem{Biernacki-al:HAL15}
D.~Biernacki, S.~Lenglet, and P.~Polesiuk.
\newblock Bisimulations for delimited-control operators.
\newblock Research report 9096, Inria Nancy -- Grand Est., Sept. 2017.
\newblock Avalaible at \url{https://hal.inria.fr/hal-01207112}.

\bibitem{Biernacki-al:MFPS17}
D.~Biernacki, S.~Lenglet, and P.~Polesiuk.
\newblock Proving soundness of extensional normal-form bisimilarities.
\newblock In A.~Silva, editor, {\em Proceedings of the 33th Annual Conference
  on Mathematical Foundations of Programming Semantics(MFPS XXXIII)}, volume
  336 of {\em Electronic Notes in Theoretical Computer Science}, pages 41--56,
  Ljubljana, Slovenia, June 2017.

\bibitem{Bird-Paterson:JFP99}
R.~S. Bird and R.~Paterson.
\newblock De {B}ruijn notation as a nested datatype.
\newblock {\em Journal of Functional Programming}, 9(1):77--91, 1999.

\bibitem{Danvy-Filinski:LFP90}
O.~Danvy and A.~Filinski.
\newblock Abstracting control.
\newblock In M.~Wand, editor, {\em LFP'90}, pages 151--160, Nice, France, June
  1990. ACM Press.

\bibitem{Dybvig-al:JFP06}
R.~K. Dybvig, S.~Peyton-Jones, and A.~Sabry.
\newblock A monadic framework for delimited continuations.
\newblock {\em Journal of Functional Programming}, 17(6):687--730, 2007.

\bibitem{Felleisen:POPL88}
M.~Felleisen.
\newblock The theory and practice of first-class prompts.
\newblock In J.~Ferrante and P.~Mager, editors, {\em POPL'88}, pages 180--190,
  San Diego, California, Jan. 1988. ACM Press.

\bibitem{Felleisen-Friedman:FDPC3}
M.~Felleisen and D.~P. Friedman.
\newblock Control operators, the {SECD} machine, and the $\lambda$-calculus.
\newblock In M.~Wirsing, editor, {\em Formal Description of Programming
  Concepts III}, pages 193--217. Elsevier Science Publishers B.V.
  (North-Holland), Amsterdam, 1986.

\bibitem{Felleisen-Hieb:TCS92}
M.~Felleisen and R.~Hieb.
\newblock The revised report on the syntactic theories of sequential control
  and state.
\newblock {\em Theoretical Computer Science}, 103(2):235--271, 1992.

\bibitem{Filinski:POPL94}
A.~Filinski.
\newblock Representing monads.
\newblock In H.-J. Boehm, editor, {\em POPL'94}, pages 446--457, Portland,
  Oregon, Jan. 1994. ACM Press.

\bibitem{Kameyama:HOSC07}
Y.~Kameyama.
\newblock Axioms for control operators in the {CPS} hierarchy.
\newblock {\em Higher-Order and Symbolic Computation}, 20(4):339--369, 2007.

\bibitem{Lassen:99}
S.~B. Lassen.
\newblock Bisimulation for pure untyped $\lambda\mu$-caluclus (extended
  abstract).
\newblock Unpublished note, Jan. 1999.

\bibitem{Lassen:MFPS99}
S.~B. Lassen.
\newblock Bisimulation in untyped lambda calculus: B{\"o}hm trees and
  bisimulation up to context.
\newblock In M.~M. Stephen~Brookes, Achim~Jung and A.~Scedrov, editors, {\em
  MFPS'99}, volume~20 of {\em ENTCS}, pages 346--374, New Orleans, LA, Apr.
  1999. Elsevier Science.

\bibitem{Lassen:LICS05}
S.~B. Lassen.
\newblock Eager normal form bisimulation.
\newblock In P.~Panangaden, editor, {\em LICS'05}, pages 345--354, Chicago, IL,
  June 2005. IEEE Computer Society Press.

\bibitem{Lassen:MFPS05}
S.~B. Lassen.
\newblock Normal form simulation for {M}c{C}arthy's amb.
\newblock In M.~Escard\'o, A.~Jung, and M.~Mislove, editors, {\em MFPS'05},
  volume 155 of {\em ENTCS}, pages 445--465, Birmingham, UK, May 2005. Elsevier
  Science Publishers.

\bibitem{Lassen:LICS06}
S.~B. Lassen.
\newblock Head normal form bisimulation for pairs and the
  $\lambda\mu$-calculus.
\newblock In R.~Alur, editor, {\em LICS'06}, pages 297--306, Seattle, WA, Aug.
  2006. IEEE Computer Society Press.

\bibitem{Lassen-Levy:CSL07}
S.~B. Lassen and P.~B. Levy.
\newblock Typed normal form bisimulation.
\newblock In J.~Duparc and T.~A. Henzinger, editors, {\em Computer Science
  Logic, 21st International Workshop, {CSL} 2007, 16th Annual Conference of the
  EACSL Proceedings}, volume 4646 of {\em Lecture Notes in Computer Science},
  pages 283--297, Lausanne, Switzerland, Sept. 2007. Springer.

\bibitem{Lassen-Levy:LICS08}
S.~B. Lassen and P.~B. Levy.
\newblock Typed normal form bisimulation for parametric polymorphism.
\newblock In F.~Pfenning, editor, {\em LICS'08}, pages 341--352, Pittsburgh,
  Pennsylvania, June 2008. IEEE Computer Society Press.

\bibitem{Madiot-al:CONCUR14}
J.~Madiot, D.~Pous, and D.~Sangiorgi.
\newblock Bisimulations up-to: Beyond first-order transition systems.
\newblock In P.~Baldan and D.~Gorla, editors, {\em 25th International
  Conference on Concurrency Theory}, volume 8704 of {\em Lecture Notes in
  Computer Science}, pages 93--108, Rome, Italy, Sept. 2014. Springer.

\bibitem{Madiot:PhD}
J.-M. Madiot.
\newblock {\em Higher-order languages: dualities and bisimulation
  enhancements}.
\newblock PhD thesis, Universit{\'e} de Lyon and Universit{\`a} di Bologna,
  2015.

\bibitem{JHMorris:PhD}
J.~H. Morris.
\newblock {\em Lambda Calculus Models of Programming Languages}.
\newblock PhD thesis, Massachusets Institute of Technology, 1968.

\bibitem{Parigot:LPAR92}
M.~Parigot.
\newblock $\lambda\mu$-calculus: an algorithmic interpretation of classical
  natural deduction.
\newblock In A.~Voronkov, editor, {\em LPAR'92}, number 624 in LNAI, pages
  190--201, St.~Petersburg, Russia, July 1992. Springer-Verlag.

\bibitem{Sangiorgi-Pous:11}
D.~Pous and D.~Sangiorgi.
\newblock Enhancements of the bisimulation proof method.
\newblock In D.~Sangiorgi and J.~Rutten, editors, {\em Advanced Topics in
  Bisimulation and Coinduction}, chapter~6, pages 233--289. Cambridge
  University Press, 2011.

\bibitem{Sabry-Felleisen:LFP92}
A.~Sabry and M.~Felleisen.
\newblock Reasoning about programs in continuation-passing style.
\newblock In W.~Clinger, editor, {\em Proceedings of the 1992 ACM Conference on
  Lisp and Functional Programming}, LISP Pointers, Vol.~V, No.~1, pages
  288--298, San Francisco, California, June 1992. ACM Press.

\bibitem{Sangiorgi:LICS92}
D.~Sangiorgi.
\newblock The lazy lambda calculus in a concurrency scenario.
\newblock In A.~Scedrov, editor, {\em LICS'92}, pages 102--109, Santa Cruz,
  California, June 1992. IEEE Computer Society.

\bibitem{Sangiorgi-al:TOPLAS11}
D.~Sangiorgi, N.~Kobayashi, and E.~Sumii.
\newblock Environmental bisimulations for higher-order languages.
\newblock {\em ACM Transactions on Programming Languages and Systems},
  33(1):1--69, Jan. 2011.

\bibitem{Stoevring-Lassen:POPL07}
K.~St{\o}vring and S.~B. Lassen.
\newblock A complete, co-inductive syntactic theory of sequential control and
  state.
\newblock In M.~Felleisen, editor, {\em POPL'07}, SIGPLAN Notices, Vol.~42,
  No.~1, pages 161--172, Nice, France, Jan. 2007. ACM Press.

\end{thebibliography}

\end{document}